\documentclass[journal]{IEEEtran}
\pdfoutput=1
\usepackage{amsmath}
\usepackage{amssymb}
\usepackage{amsfonts}
\usepackage{amsthm}
\usepackage{array}
\usepackage{enumerate}
\usepackage{cuted}
\newif\ifone
\newif\iftwo

\usepackage{bm}
\usepackage[caption=false,font=normalsize,labelfont=sf,textfont=sf]{subfig}
\usepackage{textcomp}
\usepackage{float}
\usepackage{stfloats}
\usepackage{verbatim}
\usepackage{graphicx}

\usepackage{silence}
\usepackage[colorlinks=true, allcolors=red]{hyperref}

\usepackage{soul}
\usepackage{color}
\usepackage{xcolor}
\usepackage{booktabs}
\usepackage[noadjust]{cite} 

\usepackage{balance}

\usepackage[utf8]{inputenc} 
\usepackage[T1]{fontenc}
\usepackage{url}
\usepackage{ifthen}
\usepackage{cite}
\usepackage{algpseudocode}
\usepackage[ruled,vlined,linesnumbered]{algorithm2e}
\allowdisplaybreaks

\newcommand{\red}[1]{\color{red}#1\color{black}}

\newcommand{\green}[1]{\color{green}#1\color{black}}
\newtheorem{theorem}{Theorem}
\newtheorem{lemma}[theorem]{Lemma}

\newtheorem{definition}{Definition} 
\newtheorem{remark}{Remark}

\newtheorem{example}{Example}

\newcommand{\smat}{\ \ }

\DeclareMathOperator{\blkdiag}{blkdiag}

\newlength{\spacelen}
\newcommand{\ws }[1]{\settowidth{\spacelen}{#1}\makebox[\spacelen]{}}
\newcommand{\wid}[3]{\settowidth{\spacelen}{#1}\makebox[\spacelen][#2]{$#3$}}
\newcommand{\wn}{\ws{$\gamma$}}
\newcommand{\rL}[2][\wn]{\wid{$\gamma L_{00}$}{l}{#1 L_{#2}}}
\newcommand{\bzero}{\mathbf{0}}

\newcommand{\Fq}{\mathbb{F}_q}

\newcommand{\bI}{\mathbf{I}}

\newcommand{\bO}{\mathbf{O}}

\newcommand{\cC}{\mathcal{C}}

\newcommand{\cF}{\mathcal{F}}

\newcommand{\cH}{\mathcal{H}}
\newcommand{\cI}{\mathcal{I}}

\newcommand{\cK}{\mathcal{K}}

\newcommand{\cR}{\mathcal{R}}

\newcommand{\response}[2][blue]{\textbf{Response:} \color{#1}#2\color{black}{}}
\twotrue\onefalse

\newif\ifAppendix
\Appendixtrue
\onecolumn
\usepackage{makecell}

\title{Constructing $(h,d)$ cooperative MSR codes with sub-packetization $(d-k+h)(d-k+1)^{\lceil n/2 \rceil}$}
\author{
    Zihao~Zhang,
    Guodong~Li,
    and Sihuang~Hu
    \thanks{
        Research partially funded by
        National Key R\&D Program of China under Grant No. 2021YFA1001000,
        National Natural Science Foundation of China under Grant No. 12001322 and 12231014,
        a Taishan scholar program of Shandong Province, and CCF-Huawei Populus Grove
        Fund. An early version of this paper is accepted by the 2024 IEEE International
        Symposium on Information Theory. (Corresponding author: Sihuang Hu.) } \thanks{
        Zihao Zhang, Guodong Li and Sihuang Hu are with 
        State Key Laboratory of Cryptography and Digital Economy Security, Shandong University, Qingdao, 266237, China,
        Key Laboratory of Cryptologic
        Technology and Information Security, Ministry of Education, Shandong
        University, Qingdao, Shandong, 266237, China and School of Cyber Science and
        Technology, Shandong University, Qingdao, Shandong, 266237, China. S. Hu is
        also with Quan Cheng Laboratory, Jinan 250103, China. Email: \{zihaozhang,
        guodongli\}@mail.sdu.edu.cn, husihuang@sdu.edu.cn } }

\newcommand{\bt}{\boxtimes}

\newcommand{\T}[3]{\Phi_{#1,#3}(#2)}
\newcommand{\Tg}[2]{\Phi_{\frac n2,#2}(#1)}

\newcommand{\Mp}[1]{\cI_{\cF}\left(#1\right)}
\newcommand{\R}[1]{\mathbf{rot}\left(#1\right)}

\begin{document}

\maketitle
%{\renewcommand{\thefootnote}{}
%    \footnotetext{
%        %\vspace{-.3in}		
%        %
%        %\noindent\rule{1.5in}{.4pt}	
%        %
%        Research partially funded by
%        National Key R\&D Program of China under Grant No. 2021YFA1001000,
%        National Natural Science Foundation of China under Grant No. 12001322 and 12231014,
%        %Shandong Provincial Natural Science Foundation under Grant No. ZR202010220025,
%        a Taishan scholar program of Shandong Province, and CCF-Huawei Populus Grove Fund.
%    }}
%\renewcommand{\thefootnote}{\arabic{footnote}}
%\setcounter{footnote}{0}

\begin{abstract}
    We address the multi-node failure repair challenges for MDS array codes. Presently, two primary models are employed for multi-node repairs: the centralized model where all failed nodes are restored in a singular data center, and the cooperative model where failed nodes acquire data from auxiliary nodes and collaborate amongst themselves for the repair process.
    This paper focuses on the cooperative model, and we provide explicit constructions of optimal MDS array codes with $d$ helper nodes under this model. The sub-packetization level of our new codes is $(d-k+h)(d-k+1)^{\lceil n/2 \rceil}$ where $h$ is the number of failed nodes, $k$ the number of information nodes, and $n$ the code length. This improves upon recent constructions by Liu \emph{et al.} (IEEE Transactions on Information Theory, Vol. 69, 2023).
\end{abstract}

\section{Introduction}

\IEEEPARstart{E}{rasure} codes are widely used in current distributed storage systems, where they enhance data robustness by adding redundancy to tolerate data node failures.
Common erasure codes include maximum distance separable (MDS) codes and locally repairable codes (LRC).
Particularly, MDS codes have garnered significant attention because they provide the maximum failure tolerance for a given amount of storage overhead.

An $(n,k,\ell)$ \emph{array code} has $k$ information coordinates and $r = n-k$
parity-check coordinates, where each coordinate is a vector in $\Fq^{\ell}$ for
some finite field $\Fq$. Formally, a (linear) $(n, k, \ell)$ array code $\cC$
can be defined by its parity-check equations, i.e., $$\cC = \{(C_0, \dots,
    C_{n-1}) : H_0 C_0 + \cdots + H_{n-1} C_{n-1} = \bzero\},$$ where each $C_i$ is
a column vector of length $\ell$ over $\Fq$, and each $H_i$ is a $r\ell \times
    \ell$ matrix over $\Fq$. We call $\cC$ an \emph{MDS array code} if any $r$ out
of its $n$ coordinates can be recovered from the other $k$ coordinates. To be
specific, let $\cF\ = \{i_1, i_2, \dots, i_r\}\subset[n]$ be the collection of
indices of $r$ failed nodes, we have $$ \sum_{i\in \cF}H_iC_i = -\sum_{i\in
        [n]\setminus\cF}H_iC_i, $$ where we use $[n]$ to denote the set $\{0,1,\dots,
    n-1\}$. Then we know that the $r$ coordinates $C_i\ (i\in \cF)$ can be
recovered from the other $k$ coordinates $C_i\ (i\in[n]\setminus\cF)$ if and
only if the square matrix $[H_{i_1} \smat H_{i_2} \smat \dots \smat H_{i_r}]$
is invertible. Equivalently, we say a set of $n$ matrices $H_0, H_1, \dots,
    H_{n-1}$ in $\Fq^{r\ell \times \ell}$ defines an $(n, k, \ell)$ MDS array code
if
\begin{equation*}
    [H_{i_1} \smat H_{i_2} \smat \cdots \smat H_{i_r}] \text{~is invertible,~for~} \{i_1, i_2, \dots,  i_r\}\subset [n].
\end{equation*}

With the emergence of large-scale distributed storage systems, the notion of
\emph{repair bandwidth} was introduced to measure the efficiency of recovering
the erasure of a single codeword coordinate. The seminal work by Dimakis
\emph{et. al.}~\cite{5550492} pointed out that we can repair a single failed
node by smaller repair bandwidths than the trivial MDS repair scheme. More
precisely, for an $(n,k,\ell)$ MDS array code, the optimal repair bandwidth for
a single node failure by downloading data from $d\ge k$ helper nodes is
\begin{equation}\label{eq:cut-set}
    \cfrac{d\ell}{d-k+1}.
\end{equation}
We call an $(n,k,\ell)$ MDS array code minimum storage regenerating (MSR) code with \emph{repair degree} $d$ if it achieves the lower bound \eqref{eq:cut-set} for the repair of any single erased coordinate from any $d$ out of $n-1$ remaining coordinates.
Please see~\cite{5961826, 6352912, 6737213, 7836332, 8006889, 7990181, 8104060, 8410002, 8437486,10024392, li2023msr} and references therein for
the constructions and studies of MSR codes.

MSR codes can efficiently recover a single failed node using the smallest
possible bandwidth. Naturally, new variants of MSR codes are adopted to handle
the case when $h>1$ nodes fail simultaneously. Under the centralized repair, a
single repair center downloads helper data from $d$ helper nodes and uses this
data to produce $h$ replacement nodes (please see~\cite{6401191, 8469091,
    7762203, 8006922, 8552670, 8437474, 8849496, 6880379, 8340062, 8494769,
    li2023new, Zhang2024NewCM} and references therein). Another scheme of repairing
multiple failed nodes simultaneously is cooperative repair, where failed nodes
acquire data from auxiliary nodes and collaborate amongst themselves for the
repair process. Notably, the cooperative model has demonstrated greater
robustness compared to its centralized counterpart, being able to deduce a
corresponding centralized model under equivalent parameters. Please refer
to~\cite{6565355, 5978920, 6847953, 10.1504/IJICOT.2016.079495, 8410934,
    8990088, 9137364, 10006831} and references therein for the results on
cooperative MSR codes.

This paper primarily focuses on the cooperative model and all subsequent
references to repair bandwidth and cut-set bounds are made within this context.

\begin{lemma}{(Cut-set bound~\cite{6565355,8410934})}
    \label{thm:cut-set bound}
    For an $(n,k,\ell)$ MDS array code, the optimal repair bandwidth for $h$ failed nodes by downloading information from $d$ helper nodes under the cooperative repair scheme is
    \begin{equation}\label{eq:coMSR_bound}
        \cfrac{h(d+h-1)\ell}{d-k+h}.
    \end{equation}
\end{lemma}

We say that an $(n,k,\ell)$ MDS array code $\cC$ is an \emph{$(h,d)$-MSR code
    under the cooperative model} if any $h$ failed nodes can be recovered from any
other $d$ helper nodes with total bandwidth achieving the lower
bound~\eqref{eq:coMSR_bound}. Note that a $(1,d)$-MSR code is just an MSR code
with repair degree $d$.

\subsection{Previous works on cooperative MSR codes}

In \cite{8410934}, Ye and Barg provided an explicit construction for
cooperative MSR codes with all admissible parameters. The sub-packetization
level of the construction in \cite{8410934} is given by
$((d-k)^{h-1}(d-k+h))^{\binom{n}{h}}$. Subsequent work has been focused on
reducing the sub-packetization of cooperative MSR codes. In \cite{8990088},
Zhang \emph{et al.} introduced a construction with optimal access property,
where $\ell=(d-k+h)^{\binom{n}{h}}$. Subsequently, in the work of Ye
\cite{9137364}, the sub-packetization was further reduced to
$(d-k+h)(d-k+1)^n$. More recently, Liu's work \cite{10006831} achieved even
lower sub-packetization for the case $d=k+1$: the sub-packetization of the new
construction is $o\cdot 2^n$ where $o$ is the largest odd number such that
$o\mid (h+1)$.

\begin{table*}[ht]
    \centering
    \begin{tabular}{|l|c|l|c|}
        \hline
        \makecell[c]{Codes}                                   & Sub-Packetization $\ell$              & \makecell[c]{Field Size $q$ } & Restrictions \\
        \hline
        Ye and Barg~2019~\cite{8410934}         & $((d-k)^{h-1}(d-k+h))^{\binom{n}{h}}$ & {$q\ge sn$}           &              \\
        \hline
        Zhang~\emph{et~al.}~2020~\cite{8990088} & $(d-k+h)^{\binom{n}{h}}$              & {$q\ge d-k+n$}        &              \\
        \hline
        Ye~2020~\cite{9137364}                  & $(d-k+h)s^n$                          & {$q\ge sn$}           &              \\
        \hline
        Liu~\emph{et~al.}~2023~\cite{10006831}  & $os^n$                                & {$q\ge sn$}           & $d=k+1$      \\
        \hline
        This paper                              & $(d-k+h)s^{\lceil n/2 \rceil}$        & {$q\ge sn+1$}         &              \\
        \hline
    \end{tabular}
    \vspace*{.1in}
    \caption{{Parameters of different constructions of $(h,d)$-cooperative MSR codes, where $s=d-k+1$ and $o$ is the largest odd number satisfying $o\mid d-k+h$.}}
    \label{tab:my_label}
\end{table*}

\subsection{Our contributions}
In this paper, we present a construction of cooperative MSR codes with all
admissible parameters $(h,d)$ and $\ell=(d-k+h)(d-k+1)^{\lceil n/2 \rceil}$.
{The basic ingredient of our approach is the recent construction of MSR codes in~\cite{li2023msr}, 
which introduced a method to design parity-check sub-matrices using the so-called kernel matrices and blow-up maps.
In this work, we divide the $n$ nodes into $n/2$ groups of size $2$, and
introduce two new types of kernel matrices and then blow up them to construct new $(1,d)$-MSR codes 
with sub-packetization $(d-k+1)^{\lceil n/2 \rceil}$.
Then, similarly to \cite{9137364}, we replicate the $(1,d)$-MSR code $d-k+h$ times obtaining an $(h,d)$-MSR code
with sub-packetization $(d-k+h)(d-k+1)^{\lceil n/2 \rceil}$.
%Another important ingredient of our construction is the newly introduced cooperative pairing matrices.
%The codeword $(C_0,C_1,\cdots,C_{n-1})$ of the $(n,k,\ell)$ array code is divided into $n/2$ groups of size $2$.
The optimal repair scheme is guaranteed by the deliberately chosen cooperative pairing matrices,
and it is quite different from that of~\cite{9137364}.
}

The rest of this paper is organized as follows: In Section~\ref{sec:Pre}, we
provide the necessary definitions and notations for our construction. In
Section~\ref{sec:Con-MDS}, we present our new construction and prove its MDS
property. In Section~\ref{sec:Rep-h}, we describe the repair scheme of our new
nodes, which achieves the optimal repair bandwidth. %In Section~\ref{sec:example}, we provide some detailed repair process cases to improve the readability of the paper.

\section{Preliminaries}\label{sec:Pre}

This section gives some necessary definitions and notations for the paper. Let $\Fq$ be a finite field of order $q$. For a positive integer $m$, we define
$[m]=\{0,1,\cdots,m-1\}$. For a positive integer $m$ and an integer $t$, we
define $$t+[m]=\{t+x:x\in[m]\},$$ and denote the vector $x_{[m]}$ on $\Fq$ as
$(x_0,x_1,\cdots,x_{m-1})$. Let $\mathbf{I}_m$ be the $m\times m$ identity
matrix on $\Fq$. For an element $x\in \Fq$ and a positive integer $t$, we
define a column vector of length $t$ as $$ L^{(t)}(x):=\left[\begin{array}{c}
            1      \\
            x      \\
            x^2    \\
            \vdots \\
            x^{t-1}
        \end{array}
        \right].$$
Assume that $s, t$ are two positive integers.
For each $i\in [s^t]$, we write
$$
    i = \sum_{z\in [t]}i_z s^z,~ i_z\in [s].
$$
Here we use $i_z$ to denote the $z$-th digit in the $t$ digits base-$s$ expansion of $i$.
To simplify notations, we need the matrix operator $\bt$ and the blow-up map introduced in~\cite{li2023msr}.
\begin{definition}\label{def:tensor}
    For a matrix $A$ and an $m\times n$ block matrix $B$  written as
    \begin{equation*}
        B = \begin{bmatrix}
            B_{0,0}   & \cdots & B_{0,n-1}   \\
            \vdots    & \ddots & \vdots      \\
            B_{m-1,0} & \cdots & B_{m-1,n-1}
        \end{bmatrix},
    \end{equation*}
    we define
    \begin{equation*}
        A\bt B := \begin{bmatrix}
            A\otimes B_{0,0}   & \cdots & A\otimes B_{0,n-1}   \\
            \vdots             & \ddots & \vdots               \\
            A\otimes B_{m-1,0} & \cdots & A\otimes B_{m-1,n-1}
        \end{bmatrix},
    \end{equation*}
    where $\otimes$ is the Kronecker product.
    Note that the result $A\bt B$ depends on how the rows and columns of $B$ are partitioned, and we will specify the partition every time we use this notation.
    If every block entry $B_{i,j}$ is a scalar over $\Fq$, we have $A\bt B=B\otimes A$.
\end{definition}

Throughout this paper, when we say that $B$ is a $m\times n$ block matrix, we
always assume that $B$ is uniformly partitioned, i.e., each block entry of $B$
is of the same size.

\begin{definition}[Blow-up]\label{def: Blow-up}
    Let $t$ be a positive integer.
    For any $a\in [t]$, we \emph{blow up} an $s\times s$ block matrix
    $$K=\begin{bmatrix}
            K_{0,0}   & \cdots & K_{0,s-1}   \\
            \vdots    & \ddots & \vdots      \\
            K_{s-1,0} & \cdots & K_{s-1,s-1}
        \end{bmatrix}$$
    to get an $s^{t}\times s^{t}$ block matrix via

    \begin{align*}
        \T{t}{K}{a} & =\mathbf{I}_{s^{t-a-1}}\otimes(\mathbf{I}_{s^a}\bt K)                                                                              \\
                    & =\mathbf{I}_{s^{t-a-1}}\otimes\begin{bmatrix}
                                                        \mathbf{I}_{s^a}\otimes K_{0,0}   & \cdots & \mathbf{I}_{s^a}\otimes K_{0,s-1}   \\
                                                        \vdots                            & \ddots & \vdots                              \\
                                                        \mathbf{I}_{s^a}\otimes K_{s-1,0} & \cdots & \mathbf{I}_{s^a}\otimes K_{s-1,s-1}
                                                    \end{bmatrix}.
    \end{align*}

\end{definition}

The following lemma shows the relationship between an $s\times s$ block matrix
$K$ and its blown-up $s^{t}\times s^{t}$ block matrix $\T{t}{K}{a}$.

\begin{lemma}\label{lem:repre}
    For $i,j\in[s^t]$, the block entry of $\Phi_{t, a}(K)$ at the $i$th block row and $j$th block column

    \begin{align*}
        \T{t}{K}{a}(i,j)=\begin{cases}
                             K(i_a,j_a) & \text{if~} i_z=j_z ~\forall z\in [t]\setminus\{a\} \\
                             \bO        & \text{otherwise},
                         \end{cases}
    \end{align*}
    where $K(i_a,j_a)$ is the block entry of $K$ at the $i_a$th block row and $j_a$th block column.
\end{lemma}
\begin{proof}
    We prove this lemma by induction. It is easy to see that the conclusion holds for the case $t=1$.
    Now assume that the conclusion holds for some positive integer $t$ and any $a\in[t]$, that is,
    \begin{align}\label{equ:repre}
        \T{t}{K}{a}(i,j)=\begin{cases}
                             K(i_a,j_a) & \text{if~} i_z=j_z ~\forall z\in [t]\setminus\{a\} \\
                             \bO        & \text{otherwise},
                         \end{cases}
    \end{align} where $i,j\in[s^t]$.

    We proceed to prove the case $t+1$. If $a=t$ then $\T{t+1}{K}{t}=\bI_{s^t}\bt
        K$, and we can verify that
    \begin{align*}
        \T{t+1}{K}{t}(i,j)=\begin{cases}
                               K(i_t,j_t) & \text{if~} i_z=j_z ~\forall z\in [t] \\
                               \bO        & \text{otherwise},
                           \end{cases}
    \end{align*} where $i,j\in[s^{t+1}]$.
    If $0\leq a\leq t-1$, then by definition $\T{t+1}{K}{a}=\bI_{s}\otimes \T{t}{K}{a}$.
    By~\eqref{equ:repre} we get
    \begin{align*}
        \T{t+1}{K}{a}(i,j)=\begin{cases}
                               K(i_a,j_a) & \text{if~} i_z=j_z ~\forall z\in [t+1]\setminus\{a\} \\
                               \bO        & \text{otherwise},
                           \end{cases}
    \end{align*} where $i,j\in[s^{t+1}]$. This concludes the proof.
\end{proof}
The following properties of blown-up matrices will be used for the repair scheme of our codes.

\begin{lemma}\label{lem:exchange}
    Let $A, B$ and $C$ be three $s\times s$ block matrices.
    If $$(\mathbf{I}_s\otimes A)(\mathbf{I}_s \bt B )=(\mathbf{I}_s \bt B )(\mathbf{I}_s\otimes C)\footnote{This condition is equivalent to $\T{2}{A}{0}\T{2}{B}{1}=\T{2}{B}{1}\T{2}{C}{0}.$}$$
    then for any positive integer $t$ and $a_0\neq a_1 \in[t]$,
    $$\T{t}{A}{a_0} \T{t}{B}{a_1}=\T{t}{B}{a_1} \T{t}{C}{a_0}.$$
\end{lemma}

\begin{proof}
    By Lemma~\ref{lem:repre}, we have
    $$
        \T{t}{A}{a_0}(u,v) = \begin{cases}
            A(u_{a_0},v_{a_0}) & \text{if~}u_i=v_i, \forall i\in[t]\backslash\{a_0\} \\

            \bO                & \text{otherwise},
        \end{cases}
    $$
    $$
        \T{t}{B}{a_1}(u,v) = \begin{cases}
            B(u_{a_1},v_{a_1}) & \text{if~}u_i=v_i,\forall i\in[t]\backslash\{a_1\} \\

            \bO                & \text{otherwise},
        \end{cases}
    $$
    and
    $$
        \T{t}{C}{a_0}(u,v) = \begin{cases}
            C(u_{a_0},v_{a_0}) & \text{if~}u_i=v_i,\forall i\in[t]\backslash\{a_0\} \\

            \bO                & \text{otherwise},
        \end{cases}
    $$
    where $u,v\in[s^t]$.
    We also regard $\T{t}{A}{a_0}\T{t}{B}{a_1}$ and $\T{t}{B}{a_1}\T{t}{C}{a_0}$ as $s^t\times s^t$ block matrices. Note that $a_0\neq a_1$.
    Then by the above, we can verify that
    \begin{align*}
          & [\T{t}{A}{a_0}\T{t}{B}{a_1}](u,v)                       \\
        = & \sum_{w\in[s^t]} \T{t}{A}{a_0}(u,w)\T{t}{B}{a_1}(w,v)   \\
        = & \begin{cases}
                A(u_{a_0},v_{a_0})B(u_{a_1},v_{a_1}) &
                \text{if~}u_i=v_i,\forall i\in[t]\backslash\{ a_0,a_1\}
                \\
                \bO                                  & \text{otherwise},
            \end{cases}
    \end{align*}
    and
    \begin{align*}
          & [\T{t}{B}{a_1}\T{t}{C}{a_0}](u,v)                       \\
        = & \begin{cases}
                B(u_{a_1},v_{a_1})C(u_{a_0},v_{a_0}) &
                \text{if~}u_i=v_i,\forall i\in[t]\backslash\{a_0,a_1\}
                \\
                \bO                                  & \text{otherwise.}
            \end{cases}
    \end{align*}
    Now we can see that
    $$\T{t}{A}{a_0} \T{t}{B}{a_1}=\T{t}{B}{a_1} \T{t}{C}{a_0}$$
    if and only if for any $(i_0,j_0),(i_1,j_1)\in [s]^2$, $$A(i_0,j_0) B(i_1,j_1)=B(i_1,j_1) C(i_0,j_0).$$
    The latter is equivalent to
    $$(\mathbf{I}_s\otimes A)(\mathbf{I}_s \bt B )=(\mathbf{I}_s \bt B )(\mathbf{I}_s\otimes C).$$
    This concludes our proof.
\end{proof}

The following result can be obtained easily by the mixed-product property of
the Kronecker product, therefore we omit its proof.
\begin{lemma}\label{lem:combining}
    Let $A$ and $B$ be two $s\times s$ block matrices.
    Then for any positive integer $t$ and $ a\in[t]$, we have
    $$\T{t}{A}{a}\T{t}{B}{a}=\T{t}{AB}{a}$$
    if $AB$ is a valid matrix product.
\end{lemma}
For reader's convenience, we collect the notations used in this paper in Table~\ref{tab:Tab of Notation}.
\begin{table}[!ht]\renewcommand{\arraystretch}{1.5}
    \centering
    {\begin{tabular}{|c|p{6cm}|}
            \hline
            Notation                       & \makecell[c]{Meaning}                                                                                                                                                                                                             \\\hline\hline
            \multicolumn{2}{|c|}{Code parameters}\\\hline
            $n$                           & code length                                                                                                                                                                                                           \\\hline
            $k$                           & code dimension                                                                                                                                                                                                        \\\hline
            $d$                           & repair degree                                                                                                                                                                                                         \\\hline
            $h$                           & the number of failed nodes                                                                                                                                                                                            \\\hline
            $r$                           & $n-k$                                                                                                                                                                                                                 \\\hline
            $s$                           & $d-k+1$                                                                                                                                                                                                               \\\hline
            $\tilde{\ell}$                & $s^{\lceil n/2 \rceil}$                                                                                                                                                                                               \\\hline
            $\ell$                        & $(d-k+h)s^{\lceil n/2 \rceil}$                                                                                                                                                                                        \\\hline
            $\widetilde{\cC}$             & the $(n, k, \tilde\ell)$ MSR code with repair degree $d$                                                                                                                                                              \\\hline
            $\cC$                         & the $(h, d)$ cooperative $(n,k,\ell)$ MSR code                                                                                                                                                                        \\\hline
            $\Fq$                         & the finite field with order $q$                                                                                                                                                                                       \\\hline\hline
            \multicolumn{2}{|c|}{Preliminaries}\\\hline
            $[m]$                         & $\{0,1,\cdots,m-1\}$                                                                                                                                                                                                  \\\hline
            $t+[m]$                       & $\{t,t+1,\cdots, t+m-1\}$                                                                                                                                                                                             \\\hline
            $x_{[m]}$                     & $(x_0,\cdots,x_{m-1})$                                                                                                                                                                                                \\\hline
            $L^{(t)}(x)$                  & $( 1 \quad x \quad \cdots \quad x^{t-1} )^T$                                                                                                                                                                          \\\hline
            $\bt$                         & block Kronecker product in Definition~\ref{def:tensor}                                                                                                                                                               \\\hline
            $\oplus, \oplus_s$            & additions of $\mod 2$ and $\mod s$                                                                                                                                                                                    \\\hline
            $\T{t}{\cdot}{a}$             & blow-up transformation in Definition~\ref{def: Blow-up}                                                                                                                                                              \\\hline
            $\mathbf{1}^{(s)}$            & a length-$s$ all-one column vector                                                                                                                                                                                    \\\hline
            $\cK^{(t)}(x_{[s]})$          & $\mathbf{1}^{(s)}\bt \left [
            L^{(t)}(x_0) \quad  L^{(t)}(x_1) \quad \cdots \quad  L^{(t)}(x_{s-1})\right ]$                                                                                                                                                                        \\\hline
            %$\R{\sum_{i=0}^{s-1}c_i x^i}$ & rotation map in~\eqref{eq:rot}
            $\R{\cdot}$ & rotation map in~\eqref{eq:rot}
            %$\begin{bmatrix}
            %                                         c_0     & c_1    & \cdots & c_{s-1} \\
            %                                         c_{s-1} & c_0    & \cdots & c_{s-2} \\
            %                                         \vdots  & \vdots & \ddots & \vdots  \\
            %                                         c_1     & c_2    & \cdots & c_0
            %                                     \end{bmatrix}$ 
            \\\hline\hline
            \multicolumn{2}{|c|}{Construction}\\\hline
            $\lambda_{[sn]},\gamma\in \Fq$& $sn+1$ elements satisfying local constraints~\eqref{eq:flr}                                                                                                                        \\\hline
            $L_i^{(t)}$                   & $L^{(t)}(\lambda_i)$                                                                                                                                                                                                  \\\hline
            $U_0, U_1, V_0, V_1$          & 4 $s\times s$ matrices \\\hline
            $a\in [\frac n2]$ & group index\\\hline
            $b\in [2]$ & in-group index\\\hline
            %$a, b$         & In this paper, $n$ nodes are separated into $n/2$ groups, each consisting of 2 nodes. We use $a \in [n/2]$ to denote the group index and $b \in [2]$ to denote the node's index within the group.                     \\\hline
            $K_{a,b}^{(t)}$               & kernel matrix in~\eqref{eq:kernel}                                                                             \\\hline
            $\widetilde H_{2a+b}$         & %$\Tg{K_{a,b}^{(r)}}{a}$, 
            the parity-check sub-matrix of $\widetilde C_{2a+b}$ in code $\widetilde \cC$                                                                                                                        \\\hline
            $ H_{2a+b}$                   & %$\bI_{d-k+h}\otimes \widetilde H_{2a+b}$, 
            the parity-check sub-matrix of $ C_{2a+b}$ in code $ \cC$                                                                                                                               \\\hline\hline
            \multicolumn{2}{|c|}{Cooperative repair}\\\hline
            $\cF$                         & the index set of $h$ failed nodes                                                                                                                                                                                \\\hline
            $\cH$                         & the index set of $d$ helper nodes                                                                                                                                                                 \\\hline
            $\cR_i^{\cF}$                 & the repair matrix for node $i\in\cF$                                                                                                                                                                                   \\\hline
            $\hat{i}$                     & the index of $i$ in $\cF$                                                                                                                                                                                             \\\hline
        \end{tabular}}
    \vspace*{.1in}
    \caption{Notations}
    \label{tab:Tab of Notation}
\end{table}
\section{Code construction and MDS property}\label{sec:Con-MDS}

Given code length $n$, dimension $k$, and repair degree $d$, we use $r = n-k$
to denote the redundancy of our code and set $s = d-k+1$. Assume that the
number of failed nodes $h$ satisfies that $k+1 \le d \le n-h.$ In this section,
we construct an $(n,k, \ell = (d-k+h)s^{\lceil n/2 \rceil})$ cooperative MSR
code with repair degree $d$ for any $h$ failed nodes. Without loss of
generality, we always assume that $2 \vert n$. Then $\ell = (d-k+h)s^{ n/2 }$
and we write $\tilde{\ell}=s^{ n/2 }$. The codeword $(C_0,C_1,\cdots,C_{n-1})$
of the $(n,k,\ell)$ array code is divided into $n/2$ groups of size $2$. We use
$a\in[n/2],b\in[2]$ to denote the group's index and the node's index within its
group, respectively. In other words, the group $a$ consists of the two nodes
$C_{2a}$ and $C_{2a+1}$.

To begin with, we select $sn$ distinct elements $\lambda_{[sn]}$ from $\Fq$ and
define the following \emph{kernel map} $$\cK^{(t)}:\Fq^s \to \Fq^{st\times
        s},$$ which maps $x_{[s]}$ to the following $s\times s$ block matrix
\begin{align*}
    \cK^{(t)}(x_{[s]}) & =\mathbf{1}^{(s)}\bt[
    L^{(t)}(x_0) \ L^{(t)}(x_1) \ \cdots \ L^{(t)}(x_{s-1})]                                                \\
                       & =\left[\begin{array}{cccc}
                                        L^{(t)}(x_0) & L^{(t)}(x_1) & \cdots & L^{(t)}(x_{s-1}) \\
                                        \vdots       & \vdots       & \ddots & \vdots           \\
                                        L^{(t)}(x_0) & L^{(t)}(x_1) & \cdots & L^{(t)}(x_{s-1}) \\
                                    \end{array}\right].
\end{align*}
where $\mathbf{1}^{(s)}$ is the all-one column vector of length $s$.
%and a polynomial is coefficientwise non-zero if it has no zero coefficient.
\begin{definition}
    We say a matrix is \emph{entrywise non-zero} if it has no zero entry.
    Given two entry-wise non-zero matrices $U,V\in\Fq^{s\times s}$, we call them \emph{cooperative pairing} matrices if
    $UV=\bI_s$.
\end{definition}
The cooperative pairing matrices will play a pivotal role in our cooperative repair scheme of Section~\ref{sec:Rep-h}.
% We will define the \emph{cooperative pairing} polynomials as two non-vanishing coefficient polynomials $F_0,F_1\in\Fq[x]/(x^s-1)$, where $\deg(F_0)=\deg(F_1)=s-1$ and $F_0F_1=1$. Here, non-vanishing coefficients polynomial means the polynomial with all non-zero coefficients. For any $b\in[2]$ and $g\in[s]$, we use $c_{b,g}$ to denote the coefficient of $x^g$ in $F_b$, i.e., $F_b=\sum c_{b,g}x^g$. Once we have found the \emph{cooperative pairing} polynomials $F_0$ and $F_1$, we can obtain the \emph{cooperative pairing} matrices: 
Now we provide a simple method to obtain cooperative pairing (circulant)
matrices. We first need the following useful map
\begin{equation}
    \label{eq:rot}
    \begin{array}{rccl}
        \R{\cdot}: & \Fq[x]/(x^s-1)           & \to     & \Fq^{s\times s}                                         \\\\
                   & \sum_{i=0}^{s-1} c_i x^i & \mapsto & \begin{bmatrix}
                                                              c_0     & c_1    & \cdots & c_{s-1} \\
                                                              c_{s-1} & c_0    & \cdots & c_{s-2} \\
                                                              \vdots  & \vdots & \ddots & \vdots  \\
                                                              c_1     & c_2    & \cdots & c_0
                                                          \end{bmatrix}
    \end{array},
\end{equation}
which maps a polynomial to a circulant matrix. Then the following lemma shows us how to find cooperative pairing (circulant) matrices.

\begin{lemma}\label{lem:cooppairmat}
    Choose some element $\gamma\in\Fq$ such that $g(\gamma) = \gamma(\gamma-1)(\gamma+s-1)(\gamma+s-2)\neq 0.$
    Set
    \begin{align*}
        F_0 & =x^{s-1}+\cdots+x+\gamma, \\F_1&=\frac{x^{s-1}+\cdots+x-(\gamma+s-2)}{-(\gamma-1)(\gamma+s-1)}.
    \end{align*}
    Then $F_0F_1 = 1$ in $\Fq[x]/(x^s-1)$ and $\mathbf{rot}(F_0), \mathbf{rot}(F_1)$ are cooperative pairing matrices.% are the \emph{cooperative pairing} polynomials. 
\end{lemma}
\begin{proof}
    As $g(\gamma)\neq 0$ we can verify that the matrices $\mathbf{rot}(F_0)$ and $\mathbf{rot}(F_1)$ are entrywise non-zero. By direct computations, we can easily check that $F_0F_1 = 1$ and $\mathbf{rot}(F_0)\mathbf{rot}(F_1)=\bI_s$.
\end{proof}
From now on we set
$$
    \begin{array}{ll}
        U_0=\bI_s,   & U_1=  \R{F_1}, \\
        V_0=\R{F_0}, & V_1=  \bI_s,
    \end{array}
$$
where $\mathbf{rot}(F_0)$ and $\mathbf{rot}(F_1)$ are defined as in Lemma~\ref{lem:cooppairmat}.
We can check that \footnote{For any integers $a$ and $b$, the operation $\oplus_s$ is defined as $a\oplus_s b=(a+b) \bmod{s}$. And we use $\oplus$ as a shorthand for $\oplus_2$.
}
$$\begin{array}{lcl}
        U_bV_b & = & \R{F_b}, \\U_bV_{b\oplus 1}&=&\bI_s,
    \end{array}$$
for all $b\in[2]$.

Now, we are ready to define the following \emph{kernel matrices}. For $a\in
    [n/2], b\in [2]$ and a positive integer $t$, we define
\begin{equation}\label{eq:kernel}
    \begin{aligned}
        K_{a,b}^{(t)} & =(V_b\otimes \mathbf{1}^{(t)})\odot \cK^{(t)}(\lambda_{s(2a+b)+[s]})
    % \\&=\left[\begin{array}{rr}
    %         L^{(t)}(\lambda_{4a+2b})          & \gamma_b L^{(t)}(\lambda_{4a+2b+1}) \\
    %         \gamma_b L^{(t)}(\lambda_{4a+2b}) & L^{(t)}(\lambda_{4a+2b+1})
    %     \end{array}\right],
    \end{aligned}
\end{equation}
where $\odot$ is the Hadamard (elementwise) product of two matrices.
Then, for a nonempty subset $B\subseteq [2]$, we define the horizontal concatenation matrix
$$
    K_{a,B}^{(t)} = [K_{a, b}^{(t)} : b\in B].
$$
Next, we blow up the kernel matrix to get
$$M_{a,b}^{(t)} = \Tg{K_{a,b}^{(t)}}{a} = \mathbf{I}_{s^{\frac{n}{2}-a-1}}\otimes(\mathbf{I}_{s^a}\bt K_{a,b}^{(t)}).$$
Similarly, we define $M_{a,B}^{(t)}$ as that of $K_{a,B}^{(t)}$.
Following that, we define
%{
\begin{align*}
    %= \det([\cK_0^{(t)}(x_{[2]}) \ |\ \cK_1^{(t)}(x_{2+[2]})])\\
      & f(x_{[2s]}, \gamma)                                                                                                                                                                                   \\
    = & \det\left[
        (V_0\otimes \mathbf{1}^{(2)})\odot\cK^{(2)}(x_{[s]})\ \ (V_1\otimes \mathbf{1}^{(2)})\odot\cK^{(2)}(x_{s+[s]})
    \right]                                                                                                                                                                                                   \\
    = & \scalebox{0.75}{$\det {\setlength{\arraycolsep}{0.05pt}\begin{bmatrix}
                                                      \gamma L^{(2)}(x_0) & L^{(2)}(x_1)        & \cdots & L^{(2)}(x_{s-1})        & L^{(2)}(x_s) &                  &        &                   \\
                                                      L^{(2)}(x_0)        & \gamma L^{(2)}(x_1) & \cdots & L^{(2)}(x_{s-1})        &              & L^{(2)}(x_{s+1}) &        &                   \\
                                                      \vdots              & \vdots              & \ddots & \vdots                  &              &                  & \ddots &                   \\
                                                      L^{(2)}(x_0)        & L^{(2)}(x_1)        & \cdots & \gamma L^{(2)}(x_{s-1}) &              &                  &        & L^{(2)}(x_{2s-1})
                                                  \end{bmatrix}}$}.
\end{align*}
%}
To guarantee the MDS property and the optimal repair scheme, we further require
the $sn$ distinct elements $\lambda_{[sn]}$ and $\gamma$ to satisfy
\begin{align}\label{eq:flr}
    g(\gamma)\cdot \Pi_{a\in [n/2]}f(\lambda_{2sa+[2s]}, \gamma)\neq 0.
\end{align}
The following result guarantees the existence of such elements in some linear fields.
    {%\color{red}
        \begin{lemma}\label{lem:linearfield}
            If $q\ge sn+1$, then in $\Fq$ we can always find an element $\gamma$ and $sn$ distinct elements $\lambda_{[sn]}$ satisfying \eqref{eq:flr}.
        \end{lemma}}
\begin{proof}
    By $k+1 \leq d \leq n-h$, we have $n\geq k+1+h\geq 3$ because of $k\geq 1$ and $h\geq 1$. Let $\omega$ be a primitive element of $\Fq$ with $q\geq sn+1$. Then we set $\lambda_i=\omega^{i}$ for $0\leq i\leq sn-1$.
    % and $\gamma=(1-\omega^s)^{-1}\omega^{s+1}$ which satisfy $g(\gamma)\neq 0$. Actually, $\gamma=(1-\omega^s)^{-1}(1-s-\omega^{s}) \in \{0,1\}$ is impossible. So under the condition of $q\geq sn+1\geq 7$, we have $\phi(q)\geq 4$ and must find a primitive element $\omega$ which satisfies $g(\gamma)\neq 0$. (If $s= 0$ in $\Fq$, we set $\gamma=(1-\omega^s)^{-1}(-1-\omega^{s})$ )
    We substitute these values and can observe that
    \begin{align*}
        f(\lambda_{2sa+[2s]},\gamma)=\omega^{2s^2a}f(\lambda_{[2s]},\gamma), ~0\leq a\leq n/2-1.
    \end{align*}
    Write
    $$P=\begin{bmatrix}
            \begin{matrix} 1&0\end{matrix}          &                                             &        &                                              \\
                                                    & \begin{matrix} 1&0\end{matrix}              &        &                                              \\
                                                    &                                             & \ddots &                                              \\
                                                    &                                             &        & \begin{matrix} 1&0\end{matrix}               \\
            \begin{matrix} -\lambda_s&1\end{matrix} &                                             &        &                                              \\
                                                    & \begin{matrix} -\lambda_{s+1}&1\end{matrix} &        &                                              \\
                                                    &                                             & \ddots &                                              \\
                                                    &                                             &        & \begin{matrix} -\lambda_{2s-1}&1\end{matrix} \\
        \end{bmatrix},
    $$
    and
    $$
        Q=\begin{bmatrix}
            \gamma(\lambda_0-\lambda_s) & \lambda_1-\lambda_s             & \cdots & \lambda_{s-1}-\lambda_s              \\
            \lambda_0-\lambda_{s+1}     & \gamma(\lambda_1-\lambda_{s+1}) & \cdots & \lambda_{s-1}-\lambda_{s+1}          \\
            \vdots                      & \vdots                          & \ddots & \vdots                               \\
            \lambda_0-\lambda_{2s-1}    & \lambda_1-\lambda_{2s-1}        & \cdots & \gamma(\lambda_{s-1}-\lambda_{2s-1})
        \end{bmatrix}.
    $$
    We can check that
    \begin{align*}
        &P\left[
            (V_0\otimes \mathbf{1}^{(2)})\odot\cK^{(2)}(x_{[s]})\ \ (V_1\otimes \mathbf{1}^{(2)})\odot\cK^{(2)}(x_{s+[s]})
            \right]\\=& \left[\begin{array}{c|c}
                \R{F_0} & \bI_s \\\hline
                Q       & \bO
            \end{array}\right].
    \end{align*}

    Hence
    \begin{align*}
          & f(\lambda_{[2s]},\gamma)=\det(P)^{-1}\det \left[\begin{array}{c|c}
                                                                    \R{F_0} & \bI_s \\\hline
                                                                    Q       & \bO
                                                                \end{array}\right]                                                                        \\
        = & (-1)^{\frac{s(s+1)}{2}}\scalebox{0.85}{$\det{\begin{bmatrix}
                                                \gamma(1-\omega^s) & \omega-\omega^s             & \cdots & \omega^{s-1}-\omega^s              \\
                                                1-\omega^{s+1}     & \gamma(\omega-\omega^{s+1}) & \cdots & \omega^{s-1}-\omega^{s+1}          \\
                                                \vdots             & \vdots                      & \ddots & \vdots                             \\
                                                1-\omega^{2s-1}    & \omega-\omega^{2s-1}        & \cdots & \gamma(\omega^{s-1}-\omega^{2s-1})
                                            \end{bmatrix}}$}
        \\[1.5em]= & (-1)^{\frac{s(s+1)}{2}}\omega^{\frac{s(s-1)}{2}} \scalebox{0.85}{$\det \begin{bmatrix}
                \gamma(1-\omega^s) & 1-\omega^{s-1}       & \cdots & 1-\omega           \\
                1-\omega^{s+1}     & \gamma(1-\omega^{s}) & \cdots & 1-\omega^{2}       \\
                \vdots             & \vdots               & \ddots & \vdots             \\
                1-\omega^{2s-1}    & 1-\omega^{2s-2}      & \cdots & \gamma(1-\omega^s) \\
            \end{bmatrix}$}.
    \end{align*}

    If we regard $f( \lambda_{[2s]},\gamma)$ as a polynomial in $\Fq[\gamma]$, then
    $\deg(f)=s$. Write $F(\gamma)=g(\gamma)f(\lambda_{[2s]},\gamma)$. Note that the
    condition~\eqref{eq:flr} is equivalent to $F(\gamma)\neq 0$. We see that
    $F(\gamma)$ is a non-zero polynomial in $\gamma$ with degree at most $s+4$. As
    $q\geq sn+1$, we can find an element in $\Fq$ such that $F(\gamma)$ is
    non-zero, and we assign it to $\gamma$. This concludes our proof.
\end{proof}
From now on, let $\Fq$ be a finite field with $q\geq sn+1$.
Then by Lemma~\ref{lem:linearfield} we can select
one element $\gamma$ and $sn$ distinct elements $\lambda_{[sn]}$ that satisfy~\eqref{eq:flr} from $\Fq$.

Now we write $L^{(t)}_i=L^{(t)}(\lambda_i)$. Then we have the following.
\begin{lemma}\label{lem:VMaB}
    Suppose that $a\in [n/2]$, $B\subseteq[2]$ is a nonempty set of size $t$.
    For any integer $m>t$, there exists an $\tilde{\ell} m\times \tilde{\ell} m$ matrix $V$ such that:
    \begin{itemize}
        \item [(i)]
              \begin{align*}
                  V M_{a,B}^{(m)} = \left[\begin{array}{c}
                                                  M_{a,B}^{(t)} \\
                                                  \bO
                                              \end{array}\right]
              \end{align*}
              where $\bO$ is the $\tilde\ell(m-t)\times \tilde\ell t$ all-zero matrix.
        \item [(ii)] For any $c\in [n/2]\setminus \{a\}, d\in [2]$,
              \begin{align*}
                  V M_{c,d}^{(m)} = \left[\begin{array}{l}
                                                  M_{c,d}^{(t)} \\
                                                  \widehat M_{c,d}^{(m-t)}
                                              \end{array}\right]
              \end{align*}
              %where the matrix $\cM_{g,h}^{(\cdot)} = \cM_{g,h}^{(\cdot)}(y_{[s]})$ and 
              where $\widehat{M}_{c,d}^{(m-t)}$ is an $\tilde\ell(m-t)\times \tilde\ell $
              matrix which is column equivalent to ${M}_{c,d}^{(m-t)}$.
              %\iffalse
        \item [(iii)]
              If $a\neq n/2-1$, for any $\lambda_{i_0},\cdots, \lambda_{i_{s-1}}\notin \{\lambda_{s(2a+b)+x}: b\in B,x\in[s]\}$,
              \begin{align*}
                    & V(\bI_{\tilde\ell/s}\bt \blkdiag(L_{i_0}^{(m)},\cdots, L_{i_{s-1}}^{(m)})                                                 \\
                  = & \left[\begin{array}{l}
                                    ~\bI_{\tilde\ell/s}\bt \blkdiag(L_{i_0}^{(t)},\cdots, L_{i_{s-1}}^{(t)}) \\
                                    (\bI_{\tilde\ell/s}\bt \blkdiag(L_{i_0}^{(m-t)},\cdots, L_{i_{s-1}}^{(m-t)}){\Lambda}
                                \end{array}\right]
              \end{align*}
              \footnote{Given matrices $A_i, i\in [s]$, $\blkdiag(A_i:i\in [s])$ is the block diagonal matrix obtained by aligning the matrices $A_i,i\in [s]$ along the diagonal.}where {$\Lambda$}~is an $\tilde\ell\times \tilde\ell$ invertible matrix.
    \end{itemize}
\end{lemma}

\begin{lemma}\label{lem:global}
    For any $z$ distinct integers $a_0,a_1,\cdots,a_{z-1}\in[n/2]$ and any $z$ nonempty subsets $B_0,B_1,\cdots,B_{z-1}\subseteq[2]$ satisfying $|B_0|+|B_1|+\cdots+|B_{z-1}|=m\leq r$, we have $$\det\begin{bmatrix}
            M_{a_0,B_0}^{(m)} & M_{a_1,B_1}^{(m)} & \cdots & M_{a_{z-1},B_{z-1}}^{(m)}
        \end{bmatrix}     \neq 0 .$$
\end{lemma}

Note that Lemmas~\ref{lem:VMaB}-\ref{lem:global} are almost the same as \cite[Lemma~3, Lemma~7]{li2023msr}.
Please refer to~\cite{li2023msr} for the omitted proof.
 Before giving the construction of our cooperative MSR code, we define an
intermediate $(n,k,\tilde \ell)$ array code
\begin{equation}\label{eq:inter_code}
    \widetilde{\cC}= \{(\widetilde C_0, \dots, \widetilde{C}_{n-1}): \sum_{i\in [n]}\widetilde H_i\widetilde C_i = \bzero, \widetilde C_i\in \Fq^{\tilde \ell}\},
\end{equation}
where %each $\tilde C_i$ is a column vector in $\Fq^{\tilde \ell}$, and 
$\widetilde H_{2a+b} = M_{a,b}^{(r)}$ for $a\in [n/2], b\in [2]$.
%The MDS property of the code $\widetilde{\cC}$ is obtained directly by the following.
Note that if we set $m = r$ in Lemma~\ref{lem:global}, then we obtain the MDS
property of the array code~\eqref{eq:inter_code}.
\begin{lemma}
    The code $\widetilde{\cC}$ in~\eqref{eq:inter_code} is an $(n,k,\tilde{\ell}=s^{n/2})$ MDS array code.
\end{lemma}

\begin{remark}\label{remark: MSR}
    The $(n,k,\tilde{\ell})$ MDS array code $\widetilde{\cC}$ in~\eqref{eq:inter_code} is in fact an MSR code with repair degree $d=s+k-1$.
    This can be proved similarly by the method of~\cite{li2023msr}.
\end{remark}

Finally, we give the construction of our cooperative MSR code as
\begin{equation}
    \label{eq:coMSR}
    \cC = \{(C_0, \dots, C_{n-1}): \sum_{i\in [n]} H_iC_i =\bzero, C_i\in \Fq^{\ell} \}
\end{equation}
where $H_i = \bI_{s+h-1}\otimes \tilde H_i$ for $i\in [n]$.
In other words, we replicate the $(1,d)$-MSR code $\widetilde{\cC}$ $s+h-1$ times, obtaining an $(h,d)$-MSR code.

\begin{lemma}
    The code $\cC$ in~\eqref{eq:coMSR} is an $(n,k,\ell)$ MDS array code.
\end{lemma}
\begin{proof}
    This follows directly from the fact that $\widetilde{\cC}$ is an MDS array code and     
    $H_i = \bI_{s+h-1}\otimes \tilde H_i$ for $i\in [n]$.
\end{proof}
{  In the following, we give a small example of our cooperative MSR code.
    \begin{example}\label{Exam: 6324}
        Let $n=6,k=3,h=2,$ and $d=4$. Then $s=d-k+1=2$, $\tilde \ell=2^{3}=8$, and $\ell=3\times 2^{3}=24$.
        Let $q=16$ and $\omega$ be a primitive element of $\mathbb{F}_{16}$.
        We set $\lambda_i=\omega^i$ for $0\leq i\leq 11$ and $\gamma=\frac{1}{1+\omega}$.
        As similar as the proof of Lemma~\ref{lem:linearfield}, we can check that $g(\gamma)=\gamma^2(\gamma-1)^2\neq 0$,
        and
        \begin{align*}
            f(\lambda_{[4]}, \gamma)
            =\omega\cdot {\det}
            \begin{bmatrix}
                1+\omega   & 1+\omega \\
                1+\omega^3 & 1+\omega
            \end{bmatrix}
            =\omega^2(\omega+1)^3
            \neq 0
        \end{align*}
        as the minimum polynomial of $\omega$ over $\mathbb{F}_2$ is of degree $4$.
        Hence all those values satisfy the condition~\eqref{eq:flr}.
        Now we have $$
            \begin{array}{lcll}
                U_0= & \begin{bmatrix}
                           1 & \\&1
                       \end{bmatrix},    & U_1= & \frac{1}{(\gamma+1)^2}\begin{bmatrix}
                                                                            \gamma & 1      \\
                                                                            1      & \gamma
                                                                        \end{bmatrix}, \\\\
                V_0= & \begin{bmatrix}
                           \gamma & 1      \\
                           1      & \gamma
                       \end{bmatrix}, & V_1= & \begin{bmatrix}
                                                   1 & \\&1
                                               \end{bmatrix},
            \end{array}
        $$
        and the parity-check sub-matrices
        
            %\begin{strip}\small
            \begin{align*}
                \hspace*{-.2in}  \widetilde H_0  & =
                \left[\begin{array}{@{\hspace*{0pt}}c@{\hspace*{0pt}}c@{\hspace*{0pt}}|@{\hspace*{0pt}}c@{\hspace*{0pt}}c@{\hspace*{0pt}}|@{\hspace*{0pt}}c@{\hspace*{0pt}}c@{\hspace*{0pt}}|@{\hspace*{0pt}}c@{\hspace*{0pt}}c@{\hspace*{0pt}}}
                              {\rL[\gamma]{0}} & {\rL{1}}         &                &                &                &                &                &                \\
                              {\rL{0}}         & {\rL[\gamma]{1}} &                &                &                &                &                &                \\
                              \hline
                                               &                  & \rL[\gamma]{0} & \rL{1}         &                &                &                &                \\
                                               &                  & \rL{0}         & \rL[\gamma]{1} &                &                &                &                \\
                              \hline
                                               &                  &                &                & \rL[\gamma]{0} & \rL{1}         &                &                \\
                                               &                  &                &                & \rL{0}         & \rL[\gamma]{1} &                &                \\
                              \hline
                                               &                  &                &                &                &                & \rL[\gamma]{0} & \rL{1}         \\
                                               &                  &                &                &                &                & \rL{0}         & \rL[\gamma]{1}
                          \end{array}\right],
                                                 & \widetilde H_1 & =
                \left[\begin{array}{@{\hspace*{0pt}}c@{\hspace*{0pt}}c@{\hspace*{0pt}}|@{\hspace*{0pt}}c@{\hspace*{0pt}}c@{\hspace*{0pt}}|@{\hspace*{0pt}}c@{\hspace*{0pt}}c@{\hspace*{0pt}}|@{\hspace*{0pt}}c@{\hspace*{0pt}}c@{\hspace*{0pt}}}
                              {\rL{2}} &          &        &        &        &        &        &        \\
                                       & {\rL{3}} &        &        &        &        &        &        \\
                              \hline
                                       &          & \rL{2} &        &        &        &        &        \\
                                       &          &        & \rL{3} &        &        &        &        \\
                              \hline
                                       &          &        &        & \rL{2} &        &        &        \\
                                       &          &        &        &        & \rL{3} &        &        \\
                              \hline
                                       &          &        &        &        &        & \rL{2} &        \\
                                       &          &        &        &        &        &        & \rL{3}
                          \end{array}\right], \\
                \hspace*{-.2in}  \widetilde  H_2 & =
                \left[\begin{array}{@{\hspace*{0pt}}c@{\hspace*{0pt}}c@{\hspace*{0pt}}|@{\hspace*{0pt}}c@{\hspace*{0pt}}c@{\hspace*{0pt}}|@{\hspace*{0pt}}c@{\hspace*{0pt}}c@{\hspace*{0pt}}|@{\hspace*{0pt}}c@{\hspace*{0pt}}c@{\hspace*{0pt}}}
                              {\rL[\gamma]{4}} &                & {\rL{5}}         &                &                &                &                &                \\
                                               & \rL[\gamma]{4} &                  & \rL{5}         &                &                &                &                \\\hline
                              {\rL{4}}         &                & {\rL[\gamma]{5}} &                &                &                &                &                \\
                                               & \rL{4}         &                  & \rL[\gamma]{5} &                &                &                &                \\
                              \hline
                                               &                &                  &                & \rL[\gamma]{4} &                & \rL{5}         &                \\
                                               &                &                  &                &                & \rL[\gamma]{4} &                & \rL{5}         \\
                              \hline
                                               &                &                  &                & \rL{4}         &                & \rL[\gamma]{5} &                \\
                                               &                &                  &                &                & \rL{4}         &                & \rL[\gamma]{5}
                          \end{array}\right],
                                                 & \widetilde H_3 & =
                \left[\begin{array}{@{\hspace*{0pt}}c@{\hspace*{0pt}}c@{\hspace*{0pt}}|@{\hspace*{0pt}}c@{\hspace*{0pt}}c@{\hspace*{0pt}}|@{\hspace*{0pt}}c@{\hspace*{0pt}}c@{\hspace*{0pt}}|@{\hspace*{0pt}}c@{\hspace*{0pt}}c@{\hspace*{0pt}}}
                              {\rL{6}} &        &          &        &        &        &        &        \\
                                       & \rL{6} &          &        &        &        &        &        \\
                              \hline
                                       &        & {\rL{7}} &        &        &        &        &        \\
                                       &        &          & \rL{7} &        &        &        &        \\
                              \hline
                                       &        &          &        & \rL{6} &        &        &        \\
                                       &        &          &        &        & \rL{6} &        &        \\
                              \hline
                                       &        &          &        &        &        & \rL{7} &        \\
                                       &        &          &        &        &        &        & \rL{7}
                          \end{array}\right], \\
                \hspace*{-.2in}  \widetilde H_4  & =
                \left[\begin{array}{@{\hspace*{0pt}}c@{\hspace*{0pt}}c@{\hspace*{0pt}}|@{\hspace*{0pt}}c@{\hspace*{0pt}}c@{\hspace*{0pt}}|@{\hspace*{0pt}}c@{\hspace*{0pt}}c@{\hspace*{0pt}}|@{\hspace*{0pt}}c@{\hspace*{0pt}}c@{\hspace*{0pt}}}
                              {\rL[\gamma]{8}} &                &                &                & {\rL{9}}         &                &                &                \\
                                               & \rL[\gamma]{8} &                &                &                  & \rL{9}         &                &                \\
                              \hline
                                               &                & \rL[\gamma]{8} &                &                  &                & \rL{9}         &                \\
                                               &                &                & \rL[\gamma]{8} &                  &                &                & \rL{9}         \\
                              \hline
                              {\rL{8}}         &                &                &                & {\rL[\gamma]{9}} &                &                &                \\
                                               & \rL{8}         &                &                &                  & \rL[\gamma]{9} &                &                \\
                              \hline
                                               &                & \rL{8}         &                &                  &                & \rL[\gamma]{9} &                \\
                                               &                &                & \rL{8}         &                  &                &                & \rL[\gamma]{9}
                          \end{array}\right],
                                                 & \widetilde H_5 & =
                \left[\begin{array}{@{\hspace*{0pt}}c@{\hspace*{0pt}}c@{\hspace*{0pt}}|@{\hspace*{0pt}}c@{\hspace*{0pt}}c@{\hspace*{0pt}}|@{\hspace*{0pt}}c@{\hspace*{0pt}}c@{\hspace*{0pt}}|@{\hspace*{0pt}}c@{\hspace*{0pt}}c@{\hspace*{0pt}}}
                              {\rL{10}} &         &         &         &           &         &         &         \\
                                        & \rL{10} &         &         &           &         &         &         \\
                              \hline
                                        &         & \rL{10} &         &           &         &         &         \\
                                        &         &         & \rL{10} &           &         &         &         \\
                              \hline
                                        &         &         &         & {\rL{11}} &         &         &         \\
                                        &         &         &         &           & \rL{11} &         &         \\
                              \hline
                                        &         &         &         &           &         & \rL{11} &         \\
                                        &         &         &         &           &         &         & \rL{11}
                          \end{array}\right]. 
            \end{align*}
            %\end{strip}
        The intermediate code $\widetilde \cC$ is defined as
        \begin{equation*}
            \widetilde{\cC}= \{(\widetilde C_0, \dots, \widetilde{C}_{5}): \sum_{i=0}^5\widetilde H_i\widetilde C_i = \mathbf{0}, \widetilde C_i\in \Fq^{8}\}.
        \end{equation*}
        The cooperative MSR code $\cC$ is defined as
        \begin{equation*}
            \cC = \{(C_0, \dots, C_{5}): \sum_{i=0}^5 H_iC_i =\mathbf{0}, C_i\in \Fq^{24} \}
        \end{equation*}
        where $H_i = \bI_{3}\otimes \widetilde H_i$ for $i\in \{0,\cdots,5\}$.
    \end{example}
}

\section{Repair scheme for any \texorpdfstring{$h$}{} failed nodes}\label{sec:Rep-h}
In this section, we describe the cooperative repair scheme of $\cC$ defined in~\eqref{eq:coMSR}.
Let $\cF=\{i_0, i_1, \cdots, i_{h-1}\}\subset[n]$ be the indices of any $h$ failed nodes, where $i_0 < i_1 < \cdots < i_{h-1}$.
This naturally induces a bijective map $\cI_{\cF}:\cF\to [h]$ which maps ${i_z}$ to $z$ for $z\in[h]$.
For simplicity, we write $\hat{i}=\Mp{i}$ for $i\in\cF$, i.e., $\hat{i}$ is the index of $i$ in $\cF$.
Let $\cH\subset [n]\backslash \cF$ be the collection of the indices of any $d$ helper nodes.
%\begin{remark}
%In fact, an arbitrary bijective map from $\cF$ to $[h]$ can be used to describe the repair scheme.%, a valid repair scheme can be implemented as described in the following.
%\end{remark}

%\subsection{Some useful matrices}
For $a\in[n/2], g\in[s]$, we first introduce the following
$\tilde{\ell}/s\times \tilde{\ell}$ row-selection matrix
$$R_{a,g}=\bI_{s^{n/2-a-1}}\otimes \bm e_g\otimes \bI_{s^a}$$ where $\bm e_g$
is the $g$-th row of $\bI_s$. Multiplying an $\tilde\ell\times \tilde\ell$
matrix $M$ from the left by $R_{a,g}$ is equivalent to selecting those rows in
$M$ whose indices $i$ satisfy that $i_a = g$.
%Each $\tilde\ell\times \tilde\ell$ matrix $M$ multiplies $R_{a,g}$ on the left will occur an $\tilde\ell /2\times \tilde \ell$ matrix whose rows are exactly the rows in $M$ of indices $\{i\in [\tilde \ell]: i_a = g\}$.
We can verify that
\begin{align}\label{eq:R}
    \sum_{g\in[s]}R_{a,g}^{T}R_{a,g}=\bI_{\tilde{l}}.
\end{align}

Then, for $a\in [n/2]$, $g\in[s]$ and $z\in[h]$, we define the following $s
    \times (s+h-1)$ block matrix
\begin{align}\label{def:S}
    S_{a,g,z}(i,j)=\begin{cases}
                       R_{a, g \oplus_s i } & \text{if~} j=i \text{~or~}j=z+s \\
                       \bO                  & \text{otherwise,}
                   \end{cases}
\end{align}
where $i\in[s]$, $j\in[s+h-1]$.
Note that for $z=h-1$, the case $j=z+s$ is impossible. Simply put, for $z\in[h-1]$,
\begin{align}\label{equ:S_agz}
    \scalebox{0.72}{$\begin{array}{cc}\setlength{\arraycolsep}{0.5pt}
                  & \hspace{1.8in}(z+s)\text{-th~block~column}                        \\
                  & \hspace{1.8in}\downarrow                                     \\
        S_{a,g,z} & ={\left[\begin{matrix}
                                    R_{a,g\oplus_s 0} &\cdots&\bO &\bO&\cdots&\bO &R_{a,g\oplus_s 0} &\bO&\cdots&\bO \\[.1em]
                                    \vdots&\ddots&\vdots &\vdots&\ddots&\vdots &\vdots &\vdots&\ddots&\vdots \\[.9em]
                                    \bO&\cdots &R_{a,g\oplus_s (s-1)} &\bO&\cdots&\bO &R_{a,g\oplus_s (s-1)} &\bO&\cdots&\bO\\
                                \end{matrix}\right]} \\
    \end{array}$}
\end{align}
and for $z=h-1$,
\begin{align}\label{equ:S_agh-1}
    \scalebox{0.9}{$S_{a,g,h-1}=\begin{bmatrix}
                    R_{a,g\oplus_s 0} & \cdots & \bO                   & \bO    & \cdots & \bO    \\
                    \vdots            & \ddots & \vdots                & \vdots & \ddots & \vdots \\
                    \bO               & \cdots & R_{a,g\oplus_s (s-1)} & \bO    & \cdots & \bO
                \end{bmatrix}$}.
\end{align}
{Given any matrix $M$ with $\ell$ rows, we regard $M=[M_0^T, M_1^T, \cdots, M_{s+h-2}^T]^T$ as an $(s+h-1)\times 1$ block matrix formed by vertically joining the $s+h-1$ matrices $M_i$, where each $M_i$ is a matrix with $\tilde{\ell}$ rows. Then multiplying $M$ from the left by $S_{a,g,h-1}$ is equivalent to selecting those rows in $M_i$ whose indices $j$ satisfy that $j_a = g\oplus_s i$ for $i\in[s]$, i.e.,
    \begin{align*}
        S_{a,g,h-1}M=
        \begin{bmatrix}
            R_{a,g\oplus_s 0}M_0 \\
            R_{a,g\oplus_s 1}M_1 \\
            \vdots               \\
            R_{a,g\oplus_s (s-1)}M_{s-1}
        \end{bmatrix}.
    \end{align*}
    Similarly, for $z\in[h-1]$ we have%multiplying $M$ from the left by $S_{a,g,z}$ is equivalent to 
    \begin{align*}
        S_{a,g,z}M=
        \begin{bmatrix}
            R_{a,g\oplus_s 0}(M_0+M_{z+s}) \\
            R_{a,g\oplus_s 1}(M_1+M_{z+s}) \\
            \vdots                         \\
            R_{a,g\oplus_s (s-1)}(M_{s-1}+M_{z+s})
        \end{bmatrix}.
    \end{align*}
}

%\subsubsection{Repair scheme}
For any failed node $i\in\cF$, we define the following $\tilde{\ell}\times
    (s+h-1)\tilde{\ell}$ \emph{repair matrix} $$\cR^{\cF}_{i}=S_{\lfloor \frac
    i2\rfloor,0,\hat{i}}(\bI_{s+h-1}\otimes\Tg{U_{i\bmod 2}}{\lfloor \frac
        i2\rfloor}).$$ Note that $\lfloor \frac i2\rfloor$ is the group's index of node
$i$, and $\hat{i}$ is the index of $i$ in $\cF$.
{To repair the failed
    nodes in $\cF=\{i_0, i_1, \cdots, i_{h-1}\}$, we will use the following $h$
    equations
    \begin{align*}
        (\cR_i^{\cF}\otimes \bI_r)\sum_{j\in[n]} H_j C_j=0,\ i\in\cF.
    \end{align*}
Now the cooperative pairing matrix comes into play, and we give a sketch of our proof below.
We fix some $i\in\cF$ and write $i=2a+b$, where $a=\lfloor \frac{i}{2} \rfloor$ and $b= i \bmod{2}.$
Using the fact that $U_bV_b=\R{F_b}$,
we can first compute that 
$$(\cR^{\cF}_{i}\otimes \bI_r) H_{i}=(S_{a,0,\hat{i}}\otimes\bI_r)(\bI_{s+h-1}\otimes\Tg{K}{a})$$
where $K=(\R{F_b}\otimes \mathbf{1}^{(r)})\odot \cK^{(r)}(\lambda_{si+[s]})$.
Recall that $\mathbf{rot}(F_b)$ is a circulate entry-wise nonzero matrix.
Hence by \eqref{eq:Kii} in Appendix~\ref{app:proof small-MDS} we can split the term $(\cR^{\cF}_{i}\otimes \bI_r) H_{i}C_i$ into $s$ terms, and each term carries $\frac{\ell}{s+h-1}$ ``symbols'' of $C_i$.
For the other node $j=2a+(b\oplus 1)$ in group $a$,
by $U_b V_{b\oplus 1}=\bI_s$ we compute that
$$(\cR^{\cF}_{i}\otimes \bI_r) H_{j}=(S_{a,0,\hat{i}}\otimes\bI_r)(\bI_{s+h-1}\otimes\Tg{K}{a})$$
where $K=(\bI_s\otimes \mathbf{1}^{(r)})\odot \ \cK^{(r)}(\lambda_{sj+[s]}).$
Then by \eqref{eq:Kij} in Appendix~\ref{app:proof small-MDS} the term $(\cR^{\cF}_{i}\otimes \bI_r) H_{j}C_j$ will be only transformed into
one term carrying $\frac{\ell}{s+h-1}$ ``symbols'' of $C_j$.
For the nodes $j$ does not lie in the group $a$, the term $(\cR^{\cF}_{i}\otimes \bI_r) H_{j}C_j$ will also be only transformed into
one term carrying $\frac{\ell}{s+h-1}$ ``symbols'' of $C_j$.
Collecting these $s+1+n-2$ terms together, we can define an $(n+s-1, d , \tilde\ell)$ MDS array code by Lemma~\ref{lem:small MDS},
and this enables us to recover $\frac{s\ell}{s+h-1}$ ``symbols'' of $C_i$ and $\frac{\ell}{s+h-1}$ ``symbols'' of each other node $C_j, j\in\cF$.
%For each node $i\in\cF$, the corresponding equation can recover $\frac{s\ell}{s+h-1}$ symbols of itself and $\frac{\ell}{s+h-1}$ symbols of each other node in $\cF$. 
We repeat this repair process for the failed nodes in $\cF$, and combining these ``symbols" together will complete the repair.
The details of the repair scheme are given in Lemma~\ref{lem:invertible} and Algorithm~\ref{algo:repair}.    
}
    %Each pair of failed nodes exchanges the symbols they have about one another. And then they can all recover the prime symbols about themselves. 
    %We provide Algorithm~\ref{algo:repair} to clearly outline the process.

We set the following notations for our formal statement.
\begin{itemize}
    \item [(1)] For $g\in [s]$, we define
          \begin{align}\label{def:H,C1}
              \begin{split}
                  H_{i,i}^{\langle g\rangle} & =(\cR^{\cF}_{i}\otimes \bI_r) H_{i} S^T_{\lfloor \frac  i2\rfloor,g,h-1},\\ D_{i,i}^{\langle g\rangle}&=S_{\lfloor \frac  i2\rfloor,g,\hat{i}},\\
                  C_{i,i}^{\langle g\rangle} & =D_{i,i}^{\langle g\rangle}C_i.
              \end{split}
          \end{align}
    \item [(2)] For $j\in [n]\setminus\{i\}$ with $\lfloor \frac  j2\rfloor=\lfloor \frac  i2\rfloor$,  we define
          \begin{align}\label{def:H,C2}
              \begin{split}
                  H_{i,j} & =(\cR^{\cF}_{i}\otimes \bI_r) H_{j} S^T_{\lfloor \frac  i2\rfloor,0,h-1}, \\
                  D_{i,j} & =S_{\lfloor \frac  i2\rfloor,0,\hat{i}},                                  \\
                  C_{i,j} & = D_{i,j}C_j.
              \end{split}
          \end{align}
    \item [(3)]  For $j\in [n]\setminus\{i\}$ with $\lfloor \frac  j2\rfloor \neq \lfloor \frac  i2\rfloor$, we define
          \begin{align}\label{def:H,C3}
              \begin{split}
                  H_{i,j} & =(S_{\lfloor \frac  i2\rfloor,0,\hat{i}}\otimes \bI_r)H_{j} S^T_{\lfloor \frac  i2\rfloor,0,h-1},\\ D_{i,j}&=\cR^{\cF}_{i},\\
                  C_{i,j} & = D_{i,j}C_j.
              \end{split}
          \end{align}
\end{itemize}

{At this time, $H_{i, i}^{\langle g\rangle}$ and $H_{i,j}$ are parity-check sub-matrices for the smaller code induced by the repair matrix $\cR^{\cF}_{i}$, sized $\tilde{\ell} \times \tilde{\ell}$. Meanwhile, $D_{i, i}^{\langle g\rangle}$ and $D_{i,j}$ are $\tilde{\ell} \times (s + h - 1)\tilde{\ell}$ matrices that define the codeword relation between the smaller code and $\cC$. }

%{As mentioned above, the cooperative paring matrices make two nodes in the same group "equivalent". In other words, when we repair node $i$, the other node in its group is "special" compared to the other group's.
%And this specificity is evident in the distinction between notations \eqref{def:H,C2} and \eqref{def:H,C3}.}

The following Lemmas~\ref{lem:small MDS}-\ref{lem:invertible} will be used in
the repair scheme and their proofs can be found in the Appendices.
\begin{lemma}\label{lem:small MDS}
    For each $i\in\cF$, the following $n+s-1$ matrices
    $$H_{i,0},~ \cdots,~ H_{i,i-1},~ H^{\langle 0\rangle}_{i,i},~\cdots,~ H^{\langle s-1\rangle}_{i,i},~ H_{i,i+1},~ \cdots,~ H_{i,n-1}$$
    define an $(n+s-1, d , \tilde\ell)$ MDS array code.
    And for every codeword $(C_0, \dots, C_{n-1})\in\cC$ the corresponding vector
    $$(C_{i,0},\cdots,C_{i,i-1},C^{\langle 0\rangle}_{i,i},\cdots,C^{\langle s-1\rangle}_{i,i},C_{i,i+1},\cdots,C_{i,n-1})$$
    satisfies
    $$\sum_{g\in[s]}H_{i,i}^{\langle g\rangle}C_{i,i}^{\langle g\rangle}+\sum_{j\in[n]\backslash\{i\}}H_{i,j}C_{i,j}=\mathbf{0}.$$
\end{lemma}

\begin{lemma}\label{lem:invertible}
    The $\ell\times \ell$ matrix formed by vertically joining the $s+h-1$ matrices $D_{i,i}^{\langle g\rangle},g\in[s],D_{j,i}, j\in\cF\backslash\{i\}$, is invertible.
\end{lemma}
\begin{figure}[t]
    \centering
    \includegraphics[scale=0.80]{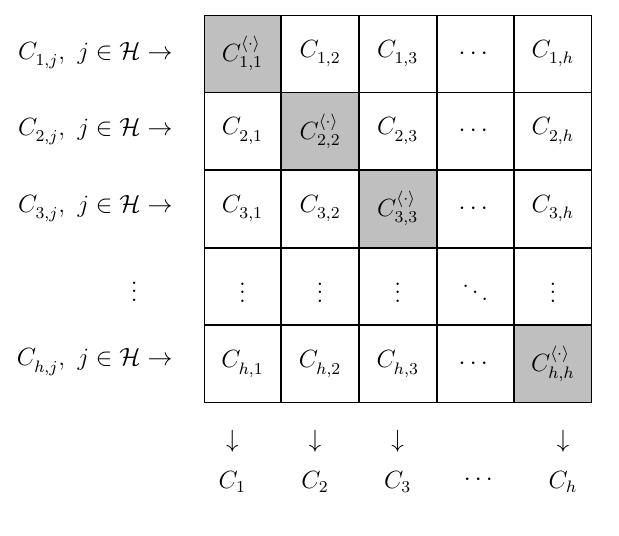}
    \caption{The repair scheme of our cooperative MSR codes.
    Without loosing of generality, we assume that $\cF = \{1,2,\dots, h\}$, and $\cH\subseteq[n]\setminus\cF$.
    For each $i\in \cF$, we have
    $C_{i,i}^{\langle g \rangle} = D_{i,i}^{\langle g \rangle}C_i, g\in [s]$, and $C_{i,j} = D_{i,j}C_j, j\in [n]\setminus\{i\}$.
    Here, for each $i\in \cF$, we use $C_{i,i}^{\langle\cdot\rangle}$ to denote the $s$ nodes $C_{i,i}^{\langle 0\rangle}, \cdots,C_{i,i}^{\langle s-1\rangle}$.
    %Note that all these computed nodes have length $\tilde \ell$ and the (gray) nodes on the main diagonal will not be transmitted again, thus they will not be counted in the repair bandwidth.
    All the off-diagonal nodes at the $i$th column will be transmitted to the node
    $C_i$. } \label{fig:enter-label}
\end{figure}

\begin{algorithm}[t]
    \DontPrintSemicolon
    \caption{\texttt{repair$(\cF, \cH)$}}
    \label{algo:repair}
    \KwIn{Two subsets $\cF, \cH\subseteq [n]$ of size $|\cF| = h, |\cH| = d$ and $\cF\cap \cH=\emptyset$, which collect the indices of failed nodes and the indices of helper nodes respectively.}
    \KwOut{The repaired nodes $\{C_i, i\in \cF\}$}
    \For{$i\in \cF$}{\label{line:1s}
    \For{$j\in \cH$}{
    Node $j$ computes $C_{i,j}=D_{i,j}C_j$

    Node $j$ \textbf{transmits} $C_{i,j}$ to node $i$\label{line:trans1} } Node $i$
    computes $$\{C_{i,i}^{\langle g \rangle}, g\in [s], C_{i,j},j\in
        \cF\setminus\{i\}\}$$ from the received data $\{C_{i,j}, j\in \cH\}$
    \Comment{Lemma~\ref{lem:small MDS}}

    \label{line:1e}
    }

    \For{$i\in \cF$}{\label{line:2s}
        \For{$j\in\cF\setminus\{i\}$}{
            Node $j$ \textbf{transmits} $C_{j,i}$ to node $i$\label{line:trans2}
        }

        Node $i$ repairs $C_i$ from $$\{C_{i,i}^{\langle g \rangle}, g\in [s], C_{j,i},
            j\in \cF\setminus\{i\}\}$$ \Comment{Lemma~\ref{lem:invertible}} \label{line:2e}
    }

    \Return $\{C_i, i\in \cF\}$
\end{algorithm}

\noindent{\textbf{Repair scheme.}} %We use $i\in\cF$ to denote the index of each failed node, and $j\in [n]\setminus\{i\}$ to denote the indices of the left nodes.
We illustrate the repair scheme in Fig.~\ref{fig:enter-label} and provide the complete steps in Algorithm~\ref{algo:repair}.
The repair process is divided into the following two steps.

%\begin{itemize}
\textbf{Step 1.}  (Row perspective of Fig.~\ref{fig:enter-label}) For each $i\in \cF$, the following steps are executed: Firstly, each helper node $j\in \cH$ calculates a vector $C_{i,j} = D_{i,j}C_j$ of length $\tilde \ell$ and sends it to node $i$. Then, by Lemma~\ref{lem:small MDS}, node $i$ can use the received data $\{C_{i,j}, j\in \cH\}$ to compute the $s+h-1$ vectors of length $\tilde\ell$, $\{C_{i,i}^{\langle g\rangle}, g\in [s], C_{i,j}, j\in \cF\setminus\{i\}\}$. These operations correspond to Lines~\ref{line:1s}-\ref{line:1e} in Algorithm~\ref{algo:repair}.

\textbf{Step 2.} (Column perspective of Fig.~\ref{fig:enter-label}) For each $i\in\cF$, node $i$ can be repaired by the following steps: First,  each node $j\in \cF\setminus\{i\}$ transmits the length-$\tilde\ell$ column vector $C_{j, i}$ computed in Step 1 to node $i$.
Recall that $$C_{i,i}^{\langle g\rangle} = D_{i,i}^{\langle g\rangle}C_i, g\in [s], C_{j,i} = D_{j,i}C_i, j\in \cF\setminus\{i\}.$$
By Lemma~\ref{lem:invertible}, $C_i$ can be recovered from $C_{i,i}^{\langle g\rangle}$, $g\in[s]$, and the received data $\{C_{j,i}, j\in \cF\setminus\{i\}\}$ from other failed nodes.
These operations correspond to Lines~\ref{line:2s}-\ref{line:2e} in Algorithm~\ref{algo:repair}.
%\end{itemize}

It is easy to check that the repair scheme achieves the lower bound of repair
bandwidth in Theorem~\ref{thm:cut-set bound}. Specifically, the length of each
intermediate vector computed during the repair process is $\tilde\ell =
    \ell/(d-k+h)$, and the steps that occupy bandwidth only occur in
Line~\ref{line:trans1} and Line~\ref{line:trans2} of Algorithm~\ref{algo:repair}. It can be
easily calculated that the bandwidth consumed during the repair process is
$$\cfrac{hd\ell}{d-k+h} + \cfrac{h(h-1)\ell}{d-k+h}.$$ Here, the left side
represents the bandwidth between failed and survival nodes, while the right
side represents the bandwidth within the $h$ failed nodes.
\begin{theorem}
    The code $\cC$ given in \eqref{eq:coMSR} is an $(h,d)$ cooperative MSR code with sub-packetization $\ell=(d-k+h)s^{\lceil n/2 \rceil}$.
\end{theorem}
\section{Conclusion}
In this paper, we construct new cooperative MSR codes for any $h$ failed nodes
and $d$ helper nodes. The sub-packetization level of our new codes is
$(d-k+h)(d-k+1)^{\lceil n/2 \rceil}$. We first construct the
$(n,k,\tilde{\ell})$ MDS array code $\widetilde{\cC}$ in~\eqref{eq:inter_code}
and then replicate $\widetilde{\cC}$ $(d-k+h)$ times, obtaining an $(h,d)$
cooperative MSR code. In general, for any collection of the number of failed
nodes $\{h_1,\cdots,h_t\}$, we can replicate $\widetilde{\cC}$
lcm$(d-k+h_1,d-k+h_2,\cdots,d-k+h_t)$ times, obtaining a new cooperative MSR
code which can repair any $h\in\{1,h_1,\cdots,h_t\}$ failed nodes with any $d$
helper nodes and the least possible bandwidth. Furthermore, the
sub-packetization of this new code is
lcm$(d-k+h_1,d-k+h_2,\cdots,d-k+h_t)(d-k+1)^{\lceil n/2 \rceil}$.

\ifAppendix
    \appendices
\fi

\section{proof of Lemma~\ref{lem:small MDS} }
\label{app:proof small-MDS}
The results of Lemma~\ref{lem:small MDS} can be divided into the following two
lemmas.

\begin{lemma}\label{lemma:repair_MDS_mat}
    For each $i\in \cF$, the $n+s-1$ matrices of size $r\tilde\ell \times \tilde\ell$,
    \begin{equation*}
        H_{i,0}, \dots, H_{i,i-1}, H_{i,i}^{\langle 0\rangle},\cdots, H_{i,i}^{\langle s-1\rangle}, H_{i,i+1}, \dots, H_{i,n-1}
    \end{equation*}
    defines an $(n+s-1, d, \tilde\ell)$ MDS array code.
\end{lemma}
\begin{lemma}\label{lemma:repair_equation}
    For $(C_0, \dots, C_{n-1})\in \cC$, we have
    \begin{align*}
          & (R_i^{\cF}\otimes I_r)(\sum_{j\in[n]}H_j C_j)                                                                            \\
        = & \sum_{g\in[s]}H_{i,i}^{\langle g\rangle}C_{i,i}^{\langle g\rangle}+\sum_{j\in[n]\backslash\{i\}}H_{i,j}C_{i,j} = \bzero.
    \end{align*}
\end{lemma}
\iffalse
    \begin{lemma}
        The dimension of the array code $\cC_i$ is $k+s-1$.
    \end{lemma}
\fi
%Here, Lemma~\ref{lemma:repair_MDS_mat}
%proves the first part of Lemma~\ref{lem:small MDS} while Lemma~\ref{lemma:repair_equation} proves the second part of Lemma~\ref{lem:small MDS}.
%   Combining these two lemmas, we complete our proof. 

%After applying the same reasoning as the proof provided in \cite[Lemma~4]{li2023msr}, we can draw the following conclusion:
% Let $$
%        K = \begin{bmatrix}
%            K_{0,0} & K_{0,1} \\
%            K_{1,0} & K_{1,1} \\
%        \end{bmatrix}
%    $$ be a $2\times 2$ block matrix in which each block entry is a column vector of length $r$, then we have the following two lemmas.
%    And, Lemma~\ref{lem:SHS=?} will be used to prove both Lemma~\ref{lemma:repair_equation} and Lemma~\ref{lemma:repair_MDS_mat}.

We first need the following technical lemma. The proof of it is exactly the
same as that of \cite[Lemma 4]{li2023msr}, and so we omit its proof. Let $$ K = \begin{bmatrix}
        K_{0,0}   & \cdots & K_{0,s-1}   \\
        \vdots    & \ddots & \vdots      \\
        K_{s-1,0} & \cdots & K_{s-1,s-1} \\
    \end{bmatrix}
$$ be a $s\times s$ block matrix in which each block entry is a column vector of length $r$.

\begin{lemma}\label{lem:14}
    For any $a ,c\in [n/2]$, $b, z \in [s]$, we have
    \begin{enumerate}
        \item[(i)] If $c=a$,
            $$(R_{a,b}\otimes\bI_r)\Tg{K}{c}R_{a,z}=\bI_{s^{\tilde c}}\otimes K_{b,z}.$$
        \item[(ii)] If $c\neq a$,
            $$(R_{a,b}\otimes\bI_r)\Tg{K}{c}R_{a,z}=\begin{cases}
                    \T{\frac n2-1}{K}{\tilde{c}} & \text{if }b=z     \\
                    \bO                          & \text{otherwise}.
                \end{cases}$$
    \end{enumerate}
    Here $$\tilde{c}=\begin{cases}
            c          & \text{if~} c<a  \\
            \frac n2-1 & \text{if~} c=a  \\
            c-1        & \text{if~} c>a. \\

        \end{cases}$$
\end{lemma}

The following result follows directly from the above.
\begin{lemma}\label{lem:SHS=?}
    For $a,c\in[n/2]$, and $z\in [h]$, we have
    %$$H=\bI_{s+h-1}\otimes \Tg{K}{a_1}$$
    %For any $g\in[s]$, we set$$K_g=\diag(K(0,g),K(1,g\oplus 1)).$$
    \begin{align*}
          & (S_{a,0,z}\otimes\bI_r)(\bI_{s+h-1}\otimes \Tg{K}{c}) S_{a,g,h-1}^T                                                  \\
        = & \begin{cases}
                \Tg{\blkdiag(K_{i,g\oplus_s i}:i\in[s])}{\tilde{c}} & \text{if~}a=c               \\
                \Tg{K}{\tilde{c}}                                   & \text{if~}a\neq c, g=0      \\
                \bO                                                 & \text{if~}a\neq c, g\neq 0,
            \end{cases}
    \end{align*}
    where $\tilde c$ is defined in Lemma~\ref{lem:14}.
\end{lemma}
\begin{proof}
    By \eqref{equ:S_agz}-\eqref{equ:S_agh-1} we can compute that
    \begin{align*}
         & (S_{a,0,z}\otimes\bI_r)(\bI_{s+h-1}\otimes \Tg{K}{c}) S_{a,g,h-1}^T
        \\= & \blkdiag\left((R_{a,i}\otimes\bI_r)\Tg{K}{c}R_{a,g\oplus_s i}):i\in[s]\right)
          % &\begin{bmatrix}
          %     (R_{a,0}\otimes\bI_r)\Tg{K}{c}R_{a,g}&&\\
          %     &\ddots &\\
          %     &&(R_{a,s-1}\otimes\bI_r)\Tg{K}{c}R_{a,g\oplus_s (s-1)})
          % \end{bmatrix}
          .
    \end{align*}
    The rest follows directly from Lemma~\ref{lem:14}.
    %Using the above lemma in $(R_{a,0}\otimes\bI_r)\Tg{K}{c}R_{a,g}$ and $(R_{a,1}\otimes\bI_r)\Tg{K}{c}R_{a,g\oplus 1}$, we finish the proof.
\end{proof}
\subsection{Proof of Lemma~\ref{lemma:repair_MDS_mat}}

%\begin{proof}
%    Considering each matrix as a block matrix divided in the same way as the corresponding definition, this conclusion can be verified directly by block matrix multiplication.  
%\end{proof}

% Following that, we give the following lemma.

% To begin with, we give a lemma that plays a crucial role in this proof.

% Let $i=2a+b$ where $a\in[n/2]$, $b\in[2]$.
% We divide the proof of this Lemma into two parts: 
% \begin{enumerate}
%     \item the MDS property of $(s+h-1)$ matrices,
%     \item $\sum_{g\in[s]}H_{i,i}^{\langle g\rangle}C_{i,i}^{\langle g\rangle}+\sum_{j\in[n]\backslash\{i\}}H_{i,j}C_{i,j}=\mathbf{0}.$ 
% \end{enumerate}

To begin with, we fix some $i\in\cF$ and set $i=2a+b$. Therefore $a=\lfloor
    \frac{i}{2} \rfloor$ and $b= i \bmod{2}.$ We first give alternative expressions
of the $n+s-1$ matrices
\begin{align}\label{n+s-1 matrices}
    H_{i,0},\cdots,H_{i,i-1},H^{\langle 0\rangle}_{i,i},\cdots,H^{\langle s-1
    \rangle}_{i,i},H_{i,i+1},\cdots,H_{i,n-1}.
\end{align}
For all $j\in[n]$, let
$$\widetilde{\left\lfloor \frac j2\right\rfloor}=\begin{cases}
        \left\lfloor \frac j2\right\rfloor   & \text{if~}\lfloor \frac j2\rfloor< a   \\
        \frac n2  -1              & \text{if~}\lfloor \frac j2\rfloor= a   \\
        \left\lfloor \frac j2\right\rfloor-1 & \text{if~} \lfloor \frac j2\rfloor> a.
    \end{cases}$$
\begin{enumerate}[1)]
    \item For any $g\in[s]$, by Lemma~\ref{lem:combining}, we have
          \begin{align*}
              H_{i,i}^{\langle g\rangle}= & (\cR^{\cF}_{i}\otimes \bI_r) H_{i} S^T_{a,g,h-1} \\=&(S_{a,0,\hat{i}}\otimes\bI_r)(\bI_{s+h-1}\otimes\Tg{K}{a})S^T_{a,g,h-1}
          \end{align*}
          where $K=(U_{b}\otimes \bI_r)K_{a,b}^{(r)}$.
          Then we can compute that
          \begin{align}\label{eq:Kii}
              K= & (U_b\otimes \bI_r)\left((V_b\otimes \mathbf{1}^{(r)})\odot \ \cK^{(r)}(\lambda_{si+[s]})\right)\nonumber
              \\= & (U_bV_b\otimes \mathbf{1}^{(r)})\odot \ \cK^{(r)}(\lambda_{si+[s]})\nonumber
              \\= & (\R{F_b}\otimes \mathbf{1}^{(r)})\odot \ \cK^{(r)}(\lambda_{si+[s]})
                % \\=&\left[\begin{array}{rr}
                %         (1-\gamma_0\gamma_1)L^{(r)}_{2i}             & (\gamma_{b}-\gamma_{b\oplus 1}) L^{(t)}_{2i+1} \\
                %         (\gamma_{b}-\gamma_{b\oplus 1}) L^{(r)}_{2i} & (1-\gamma_0\gamma_1)L^{(r)}_{2i+1}
                %     \end{array}\right],
          \end{align}
          where $L_i^{(r)} = L^{(r)}(\lambda_i)$.
          Using Lemma~\ref{lem:SHS=?}, we can compute that for all $g\in[s]$,
          \begin{align}\label{subrepair1}
              \begin{split}
                  H_{i,i}^{\langle g\rangle} & = c_{b,g}\Tg{\blkdiag(L_{si+(g\oplus_s x)}:x\in[s])}{\widetilde{\lfloor \frac i2\rfloor}} .            \\
              \end{split}
          \end{align}
          where $c_{b,g}$ is the coefficient of $x^{g}$ in $F_b$.
    \item For $j\in [n]\setminus\{i\}$ with $\lfloor \frac j2\rfloor=a$, we have $j\bmod
              2=b\oplus 1$ and
          \begin{align*}
              H_{i,j}= & (\cR^{\cF}_{i}\otimes \bI_r) H_{j} S^T_{a,0,h-1} \\=&(S_{a,0,\hat{i}}\otimes\bI_r)(\bI_{s+h-1}\otimes\Tg{K}{a})S^T_{a,0,h-1}
          \end{align*}
          where $K=(U_{b}\otimes \bI_r)K_{a,b\oplus 1}^{(r)}$.
          Then we can compute that
          \begin{align}\label{eq:Kij}
              K= & (U_b\otimes \bI_r)\left((V_{b\oplus 1}\otimes \mathbf{1}^{(r)})\odot \ \cK^{(r)}(\lambda_{sj+[s]})\right)\nonumber
              \\= & (U_bV_{b\oplus 1}\otimes \mathbf{1}^{(r)})\odot \ \cK^{(r)}(\lambda_{sj+[s]})\nonumber
              \\= & (\bI_s\otimes \mathbf{1}^{(r)})\odot \ \cK^{(r)}(\lambda_{sj+[s]}).
                %     \\=&\left[\begin{array}{rr}
                %   (1-\gamma_{b\oplus 1}^2)L^{(r)}_{2j} &                             \\
                %                               & (1-\gamma_{b\oplus 1}^2)L^{(r)}_{2j+1}
                %         \end{array}\right],
          \end{align}
          Using Lemma~\ref{lem:SHS=?}, we can compute that
          \begin{align}\label{subrepair2}
              H_{i,j} & =\Tg{\blkdiag(L_{sj+x}:x\in[s])}{\widetilde{\lfloor \frac j2\rfloor}} \\
                      & =\Tg{\blkdiag(L_{sj+x}:x\in[s])}{\frac{n}{2}-1}.
          \end{align}
    \item For $j\in [n]\setminus\{i\}$ with $\lfloor \frac j2\rfloor \neq a$,
          \begin{align*}
              H_{i,j}= & (S_{a,0,\hat{i}}\otimes \bI_r)H_{j} S^T_{a,0,h-1}
              \\=&(S_{a,0,\hat{i}}\otimes \bI_r)(\bI_{s+h-1}\otimes\Tg{K}{\lfloor \frac  j2\rfloor}) S^T_{a,0,h-1}.
          \end{align*}
          where $K=K_{\lfloor \frac  j2\rfloor,j\bmod 2}^{(r)}$.
          And by Lemma~\ref{lem:SHS=?}, we can directly compute that
          \begin{align}\label{subrepair3}
              H_{i,j}=\Tg{K^{(r)}_{\lfloor \frac j2\rfloor,j\bmod 2}}{\widetilde{\lfloor \frac j2\rfloor}}.
          \end{align}
\end{enumerate}

From~\eqref{subrepair1}, \eqref{subrepair2}, and \eqref{subrepair3}, we can
observe that the structure of $n+s-1$ matrices defined in~\eqref{n+s-1
    matrices} is similar to that of parity-check sub-matrices
of~\eqref{eq:inter_code}. Using Lemma~\ref{lem:VMaB} and the same approach as
in Lemma~\ref{lem:global}, we can prove Lemma~\ref{lemma:repair_MDS_mat}.

\subsection{Proof of Lemma~\ref{lemma:repair_equation}}

\begin{lemma}\label{lem:RHC=?}
    For each $i\in \cF$, we write $i = 2a+b$, where $a\in [n/2]$ and $b\in [2]$. Then for any $j\in[n]$, we have
    \begin{align*}
         & (\cR_i^{\cF}\otimes \bI_r)H_jC_j \\= &
          \begin{cases}
            \sum_{g\in[s]}\left[ (\cR^{\cF}_{i}\otimes \bI_r )H_{i}S^T_{a,g,h-1}\right](S_{a,g,\hat{i}}C_{i}) & j=i,                                    \\
            \\
            \left[ (\cR^{\cF}_{i}\otimes \bI_r )H_{j}S^T_{a,0,h-1}\right](S_{a,0,\hat{i}}C_{j})               & j\neq i, \lfloor \frac j2\rfloor=a,     \\
            \\
            \left[ (S_{a,0,\hat{i}}\otimes \bI_r)H_{j}S^T_{a,0,h-1}\right](\cR^{\cF}_{i}C_{j})                & j\neq i, \lfloor \frac j2\rfloor\neq a.
        \end{cases}
    \end{align*}

\end{lemma}

\begin{proof}
    Firstly, for $z\in[h]$, we define an $(s+h-1)\times (s+h-1)$ block matrix
    \begin{equation}
        \label{eq:Qz}
        Q_z(i,j)=\begin{cases}
            \wn\bI_{\tilde{\ell}} & \text{if~}i=j\in[s+h-1]\backslash[s] \\
            -\bI_{\tilde{\ell}}   & \text{if~}i\in[s],~j=z+s             \\
            \wn \bO               & \text{otherwise,}
        \end{cases}
    \end{equation}
    and we can see that $Q_z$ is an $\ell\times \ell$ matrix.
    Furthermore, we have the following two conclusions, which can be proved directly by  \eqref{eq:R}, \eqref{def:S} and \eqref{eq:Qz}:
    \begin{enumerate}
        \item For any $a\in[n/2]$ and $z\in[h]$,
              \begin{align}\label{eq:I}
                  \sum_{g\in[s]} S^T_{a,g,h-1}S_{a,g,z} +Q_z=\bI_{\ell}.
              \end{align}
        \item  For any $a\in[n/2]$, $z\in[h]$ and $r\tilde{\ell}\times \tilde{\ell}$ matrix
              $M$, we have
              \begin{align}\label{eq:O}
                  (S_{a,0,z}\otimes \bI_r ) (\bI_{s+h-1}\otimes M)Q_z=\bO.
              \end{align}
    \end{enumerate}
    We write $E_{a,b}=\bI_{s+h-1}\otimes\Tg{U_{b}}{a}$. Then $R_{i}^{\cF} = S_{a,0,\hat i}E_{a,b}$.
    \begin{enumerate}
        \item If $j=i$,
              \begin{align*}
                   & (\cR_i^{\cF}\otimes \bI_r)H_jC_j
                  \\= & (\cR^{\cF}_{i}\otimes \bI_r )H_{i}(\sum_{g\in[s]} S^T_{a,g,h-1}S_{a,g,\hat{i}} +Q_{\hat{i}}) C_{i}
                  \\= & \sum_{g\in[s]}\left[ (\cR^{\cF}_{i}\otimes \bI_r )H_{i}S^T_{a,g,h-1}\right](S_{a,g,\hat{i}}C_{i})
                  \\           & +(\cR^{\cF}_{i}\otimes \bI_r )H_{i}Q_{\hat{i}} C_{i}
                              % \\=&\sum_{g\in[s]}\left[ (S_{a,0,\hat{i}}\otimes\bI_r)(\bI_{s+h-1}\otimes\Tg{K}{a})S^T_{a,g,h-1}\right](S_{a,g,\hat{i}}C_{i})
                              % \\&+(S_{a,0,\hat{i}}\otimes\bI_r)(\bI_{s+h-1}\otimes\Tg{K}{a})Q_{\hat{i}} C_{i}
              \end{align*}

              By \eqref{eq:O}, we have \begin{align*}
                   & (\cR^{\cF}_{i}\otimes \bI_r )H_{i}Q_{\hat{i}} \\=&(S_{a,0,\hat{i}}\otimes\bI_r)(\bI_{s+h-1}\otimes\Tg{K}{a})Q_{\hat{i}}\\=&\bO,
              \end{align*} where
              $K=(U_{b}\otimes \bI_r)K_{a,b}^{(r)}$, computed in \eqref{eq:Kii}. Therefore,
              \begin{align*}
                   & (\cR_i^{\cF}\otimes \bI_r)H_jC_j
                  \\=&\sum_{g\in[s]}\left[ (\cR^{\cF}_{i}\otimes \bI_r )H_{i}S^T_{a,g,h-1}\right](S_{a,g,\hat{i}}C_{i}).
              \end{align*}

        \item For $j\in[n]\backslash\{i\}$ and $\lfloor j/2\rfloor= a$, Similar to the above,
              we have
              \begin{align*}
                   & (\cR_i^{\cF}\otimes \bI_r)H_jC_j
                  \\= & (\cR^{\cF}_{i}\otimes \bI_r )H_{j}(\sum_{g\in[s]} S^T_{a,g,h-1}S_{a,g,\hat{i}} +Q_{\hat{i}}) C_{j}
                  \\= & \sum_{g\in[s]}\left[ (\cR^{\cF}_{i}\otimes \bI_r )H_{j}S^T_{a,g,h-1}\right](S_{a,g,\hat{i}}C_{j})
                  \\&+(\cR^{\cF}_{i}\otimes \bI_r )H_{j}Q_{\hat{i}} C_{j}.
              \end{align*}
              Let $K=(U_{b}\otimes \bI_r)K_{a,b\oplus 1}^{(r)}$, computed in \eqref{eq:Kij}. By \eqref{eq:O}, we have \begin{align*}
                   & (\cR^{\cF}_{i}\otimes \bI_r )H_{j}Q_{\hat{i}} \\=&(S_{a,0,\hat{i}}\otimes\bI_r)(\bI_{s+h-1}\otimes\Tg{K}{a})Q_{\hat{i}}\\=&\bO.
              \end{align*}
              By Lemma~\ref{lem:SHS=?} we can get that for any $g\in[s]\backslash\{0\}$,
              \begin{align*}
                   & (\cR^{\cF}_{i}\otimes \bI_r )H_{j}S^T_{a,g,h-1}
                  \\= & (S_{a,0,\hat{i}}\otimes\bI_r)(\bI_{s+h-1}\otimes\Tg{K}{a})S^T_{a,g,h-1}
                  \\=&\bO.
              \end{align*}
              Combining the above we have
              \begin{align*}
                   & (\cR^{\cF}_{i}\otimes \bI_r )H_{j}(\sum_{g\in[s]} S^T_{a,g,h-1}S_{a,g,\hat{i}} +Q_{\hat{i}}) C_{j}
                  \\=&\left[ (\cR^{\cF}_{i}\otimes \bI_r )H_{j}S^T_{a,0,h-1}\right](S_{a,0,\hat{i}}C_{j}).
              \end{align*}

        \item For $j\in[n]\backslash\{i\}$ and $\lfloor j/2\rfloor\neq a$. Using
              Lemma~\ref{lem:exchange} directly, we have $$(E_{a,b}\otimes
                  \bI_r)H_j=H_jE_{a,b}.$$ Then
              \begin{align*}
                   & (\cR_i^{\cF}\otimes \bI_r)H_jC_j
                  \\= & (S_{a,0,\hat{i}}\otimes \bI_r) (E_{a,b}\otimes \bI_r)H_{j} C_{j}
                  \\= & (S_{a,0,\hat{i}}\otimes \bI_r) H_{j}E_{a,b} C_{j}
                  \\= & (S_{a,0,\hat{i}}\otimes \bI_r) H_{j}(\sum_{g\in[s]} S^T_{a,g,h-1}S_{a,g,\hat{i}} +Q_{\hat{i}})E_{a,b} C_{j}
                  \\= & \sum_{g\in[s]}\left[ (S_{a,0,\hat{i}}\otimes \bI_r)H_{j}S^T_{a,g,h-1}\right](S_{a,g,\hat{i}}E_{a,b}C_{j})
                  \\           & +(S_{a,0,\hat{i}}\otimes \bI_r)H_{j}Q_{\hat{i}} E_{a,b}C_{j}.
                              % \\=&\sum_{g\in[s]}\left[ (S_{a,0,\hat{i}}\otimes \bI_r)(\bI_{s+h-1}\otimes\Tg{K}{\lfloor \frac  j2\rfloor})S^T_{a,g,h-1}\right](S_{a,g,\hat{i}}E_{a,b}C_{j})
                              % \\&+(S_{a,0,\hat{i}}\otimes \bI_r)(\bI_{s+h-1}\otimes\Tg{K}{\lfloor \frac  j2\rfloor})Q_{\hat{i}} E_{a,b}C_{j}.
              \end{align*}
              Because $H_j=\bI_{s+h-1}\otimes\Tg{K_{\lfloor \frac  j2\rfloor,j\bmod 2}^{(r)}}{\lfloor \frac  j2\rfloor}$, using Lemma~\ref{lem:SHS=?} and \eqref{eq:O}, we have
              \begin{enumerate}[(i)]
                  \item for any $g\in[s]\backslash\{0\}$, $$(S_{a,0,\hat{i}}\otimes
                            \bI_r)H_{j}S^T_{a,g,h-1}=\bO.$$
                  \item $(S_{a,0,\hat{i}}\otimes \bI_r)H_{j}Q_{\hat{i}}=\bO$.
              \end{enumerate}
              Therefore, we have
              \begin{align*}
                   & (\cR^{\cF}_{i}\otimes \bI_r )H_{j}(\sum_{g\in[s]} S^T_{a,g,h-1}S_{a,g,\hat{i}} +Q_{\hat{i}}) C_{j}
                  \\=&\left[ (S_{a,0,\hat{i}}\otimes \bI_r)H_{j}S^T_{a,0,h-1}\right](S_{a,0,\hat{i}}E_{a,b}C_{j})
              \end{align*}
    \end{enumerate}
    %\fi

\end{proof}

In summary, we have
\begin{align}
      & (R_i^{\cF}\otimes \bI_r)(\sum_{j\in [n]}H_jC_j)                                                                            \\
    = & \sum_{j\in [n]}(R_i^{\cF}\otimes \bI_r)H_jC_j\label{eq:a}                                                                  \\
    = & \sum_{g\in[s]}\left[ (\cR^{\cF}_{i}\otimes \bI_r )H_{i}S^T_{a,g,h-1}\right](S_{a,g,\hat{i}}C_{i})\nonumber
    \\           & +\left[ (\cR^{\cF}_{i}\otimes \bI_r )H_{2a+(b\oplus1)}S^T_{a,0,h-1}\right](S_{a,0,\hat{i}}C_{2a+(b\oplus1)})\nonumber
    \\           & +\sum_{j\in[n]\backslash(2a+[2])}\left[ (S_{a,0,\hat{i}}\otimes \bI_r)H_{j}S^T_{a,0,h-1}\right](\cR^{\cF}_{i}C_{j})
    \label{eq:b}                                                                                                                   \\
    = & \sum_{g\in[s]}H_{i,i}^{\langle g\rangle}C_{i,i}^{\langle g\rangle}+\sum_{j\in[n]\backslash\{i\}}H_{i,j}C_{i,j}\label{eq:c}
    \\= &\bzero.
\end{align}
Using Lemma \ref{lem:RHC=?}, we can deduce \eqref{eq:b} from \eqref{eq:a}. By applying notations \eqref{def:H,C1}, \eqref{def:H,C2} and \eqref{def:H,C3}, we can transform \eqref{eq:b} to \eqref{eq:c}.

\section{proof of Lemma~\ref{lem:invertible}}
\label{app:proof-invertible}

%For any $i\in\cF$, we write $i=2a+b$ where $a\in[n/2]$ and $b\in[2]$.
For any $i,j \in \cF$ we define $$P_{j,i}=\begin{cases}
        \begin{bmatrix}
            R_{\lfloor \frac j2\rfloor,0} \\
            \vdots                        \\
            R_{\lfloor \frac j2\rfloor,s-1}
        \end{bmatrix}      & \text{if~} \lfloor \frac j2\rfloor=\lfloor \frac i2\rfloor,                                          \\\\
        \begin{bmatrix}
            R_{\lfloor \frac j2\rfloor,0} \\
            \vdots                        \\
            R_{\lfloor \frac j2\rfloor,s-1}
        \end{bmatrix}\Tg{U_{j\bmod 2}}{\lfloor \frac j2\rfloor} & \text{if~} \lfloor \frac j2\rfloor\neq \lfloor \frac i2\rfloor,
    \end{cases}$$
which are all invertible matrices.

We also define that $E_z= \epsilon_z \otimes \bI_{\tilde{\ell}}$ where
$\epsilon_z$ is the $z$-th row of $\bI_{s+h-1}$. We can easily check that the
the $\ell\times \ell$ matrix formed by vertically joining the $s+h-1$ matrices
$E_z,z\in[s+h-1]$, is invertible. For $x,y\in[s]$, set $W_{x,y}$ to be the
$s\times s$ block matrix with block entry of size $\tilde{\ell}/s $ and for all
$i,j\in[s]$,
\begin{align}
    W_{x,y}(i,j)=\begin{cases}
                     \bI_{\tilde{\ell}/s} & i=x,j=y           \\
                     \bO                  & \text{otherwise}.
                 \end{cases}
\end{align}

We now split the proof into two cases.

\textbf{Case 1:  $\hat{i}\in[h-1]$.}
We can see for all $g\in[s]$,
\begin{align*}
    \scalebox{0.72}{$\begin{array}{cc}\setlength{\arraycolsep}{0.5pt}
                                   & \hspace{1.8in}(\hat{i}+s)\text{-th~block~column}                                                  \\
                                   & \hspace{1.8in}\downarrow                                      \\
        D^{\langle g\rangle}_{i,i} & =\left[\begin{matrix}
                                                     R_{a,g\oplus_s 0} &\cdots &\bO &\bO&\cdots&\bO &R_{a,g\oplus_s 0} &\bO&\cdots&\bO \\[.1em]
                                                     \vdots&\ddots&\vdots &\vdots&\ddots&\vdots &\vdots &\vdots&\ddots&\vdots \\[.9em]
                                                     \bO&\cdots &R_{a,g\oplus_s (s-1)} &\bO&\cdots&\bO &R_{a,g\oplus_s (s-1)} &\bO&\cdots&\bO
                                                 \end{matrix}\right]. \\
    \end{array}$}
\end{align*}

By performing operations on the rows of the matrices, we can get for $z\in[s]$,
\begin{align}\label{eq:z}
    M_z:= & P_{i,i}^{-1}\left(\sum_{g\in[s]}W_{g\oplus_s z,z}D^{\langle g\rangle}_{i,i}\right)\nonumber
    =E_z+E_{\hat{i}+s}.
\end{align}
Let $k\in\cF$ be the failed node with $\hat{k}=h-1$. Then we can check that
\begin{align*}
    E_{\hat{i}+s} & =P_{k,i}^{-1}\left(
    \sum_{z\in[s]}W_{z,z}P_{k,i}M_z-D_{k,i}
    \right)
\end{align*}
and for all $z\in[s]$,
\begin{align*}
    E_z=M_z-E_{\hat{i}+s}.
\end{align*}
For any $j\in\cF\backslash\{i,k\}$, i.e. $\hat{j}\neq h-1,\hat{i}$, we can also check that
\begin{align*}
    E_{\hat{j}+s} & =P_{j,i}^{-1}\left(D_{j,i}-
    \sum_{z\in[s]}W_{z,z}P_{j,i}E_z\right).
\end{align*}
Therefore, we can see that every $E_z, z \in [s+h-1]$ can be written as a linear combination of  the $s+h-1$ matrices $D_{i,i}^{\langle g\rangle},g\in[s],D_{j,i}, j\in\cF\backslash\{i\}$. This implies that the $\ell\times \ell$ matrix formed by vertically joining the $s+h-1$ matrices, which includes $D_{i,i}^{\langle g\rangle},g\in[s],D_{j,i}, j\in\cF\backslash\{i\}$, is invertible for all $i\in\cF$ satisfying $\hat{i}\in[h-1]$.\\

\textbf{Case 2:  $\hat{i}=h-1$.}
In this case, we can see for all $g\in[s]$,
$$D^{\langle g\rangle}_{i,i}=\begin{bmatrix}
        R_{a,g\oplus_s 0} & \cdots & \bO                   & \bO    & \cdots & \bO    \\
        \vdots            & \ddots & \vdots                & \vdots & \ddots & \vdots \\
        \bO               & \cdots & R_{a,g\oplus_s (s-1)} & \bO    & \cdots & \bO
    \end{bmatrix}$$
As same as case 1, we can get for all $z\in[s]$,
\begin{align*}
    E_z= & P_{i,i}^{-1}\left(
    \sum_{g\in[s]}W_{g\oplus_s z,z}D^{\langle g\rangle}_{i,i}
    \right).
\end{align*}
And then for all $j\in\cF\backslash\{i\}$, we have
\begin{align*}
    E_{\hat{j}+2} & =P_{j,i}^{-1}\left(D_{j,i}-
    \sum_{z\in[s]}W_{z,z}P_{j,i}E_z\right).
\end{align*}
As above, we can get all $E_z$ for $z\in[s+h-1]$ by linear combination of  the $s+h-1$ matrices $D_{i,i}^{\langle g\rangle},g\in[s],D_{j,i}, j\in\cF\backslash\{i\}$ again, which means the $\ell\times \ell$ matrix formed by vertically joining the $s+h-1$ matrices $D_{i,i}^{\langle g\rangle},g\in[s],D_{j,i}, j\in\cF\backslash\{i\}$, is invertible for $\hat{i}=h-1$.
\bibliographystyle{ieeetr}
\bibliography{TIT}
\iffalse
\begin{IEEEbiographynophoto}{Zihao~Zhang}
    is currently a Ph.D. student at School of Cyber Science and Technology, Shandong University, Qingdao, China.
    His research interests include coding theory and combinatorics.
\end{IEEEbiographynophoto}

\begin{IEEEbiographynophoto}{Guodong~Li}
    is currently a Ph.D. student at School of Cyber Science and Technology, Shandong University, Qingdao, China.
    His research interests include information theory and coding theory.
    He received the 2024 Jack Keil Wolf ISIT Student Paper Award.
\end{IEEEbiographynophoto}

\begin{IEEEbiographynophoto}{Sihuang Hu}
    received the B.Sc. degree in 2008 from Beihang University, Beijing, China, and the Ph.D. degree in 2014 from Zhejiang University, Hangzhou, China, both in applied mathematics.
    He is currently a professor at Shandong University, Qingdao, China.
    Before that, he was a postdoc at RWTH Aachen University, Germany, from 2017 to 2019 and from 2014 to 2015, and a postdoc at Tel Aviv University, Israel, from 2015 to 2017.
    His research interests include lattices, combinatorics, and coding theory for communication and storage system.
    Sihuang is a recipient of the Humboldt Research Fellowship (2017).
    His coauthored paper received the 2024 Jack Keil Wolf ISIT Student Paper Award.
\end{IEEEbiographynophoto}
\fi
\vfill

\end{document}